\documentclass[sts]{imsart}

\RequirePackage{amsthm,amsmath,amsfonts,amssymb}
\RequirePackage[numbers]{natbib}
\RequirePackage[colorlinks,citecolor=blue,urlcolor=blue]{hyperref}
\RequirePackage{graphicx}
\usepackage{tikz}
\usepackage{pgfplots}
\usetikzlibrary{shapes.geometric, arrows.meta, fit, positioning, arrows, calc} 

\startlocaldefs
\numberwithin{equation}{section}
\theoremstyle{plain}
\newtheorem{theorem}{Theorem}[section]

\newtheorem{claim}{Claim}[section]
\newtheorem{corollary}{Corollary}[section]

\newtheorem{proposition}{Proposition}[section]

\newcommand{\revision}[1]{\textcolor{black}{#1}} 

\theoremstyle{remark}
\newtheorem{assumption}{Assumption}[section]
\newtheorem{definition}{Definition}[section]

\endlocaldefs

\begin{document}

\begin{frontmatter}
\title{A primer on optimal transport for causal inference with observational data}
\runtitle{A primer on OT for CI with observational data}

\begin{aug}
\author[A]{\fnms{Florian}~\snm{Gunsilius}\ead[label=e1]{fgunsil@emory.edu}\orcid{0000-0002-1698-6324} }
\address[A]{Department of Economics, Emory University\printead[presep={\ }]{e1}. Zheng Fang, Rex Hsieh, Esfandiar Maasoumi, Guido Romero, and Alejandro Sanchez-Becerra provided greatly helpful feedback. All errors are the author's.}
\end{aug}

\begin{abstract}
The theory of optimal transportation has developed into a powerful and elegant framework for comparing probability distributions, with wide-ranging applications in all areas of science. The fundamental idea of analyzing probabilities by comparing their underlying state space naturally aligns with the core idea of causal inference, where understanding and quantifying counterfactual states is paramount. Despite this intuitive connection, explicit research at the intersection of optimal transport and causal inference is only beginning to develop. Yet, many foundational models in causal inference have implicitly relied on optimal transport principles for decades, without recognizing the underlying connection.
Therefore, the goal of this review is to offer an introduction to the surprisingly deep existing connections between optimal transport and the identification of causal effects with observational data---where optimal transport is not just a set of potential tools, but actually builds the foundation of model assumptions. As a result, this review is intended to unify the language and notation between different areas of statistics, mathematics, and econometrics, by pointing out these existing connections, and to explore novel problems and directions for future work in both areas derived from this realization. 

\end{abstract}


\end{frontmatter}
\section{Introduction}
At the heart of many scientific and practical problems lies the challenge of understanding causal relationships of interventions and modifications: will I feel better if I take this medication? What are the health effects on the population of opening up a factory near this neighborhood? These questions all revolve around a fundamental principle: identifying and analyzing the effect an intervention---taking the medication, building the factory in this neighborhood---has on the system under consideration, that is, my body or the neighborhood. Such questions are difficult to answer with data since they require the knowledge of two mutually exclusive states of the world, the \emph{counterfactual states}. Either the factory is built at this location or it is not, but it is impossible to ever observe both states at the same time. This is the ``fundamental problem of causal inference'' \citep{holland1986statistics}. 

This problem is amplified in settings where the statistician does not have the agency to devise the intervention directly---as is the case in randomized controlled trials for instance---but only has access to observational data. In such cases the treatment variable is \emph{endogenous} to the system, meaning that it is not independent of the unobservables that affect the outcome of interest. A classic example is self-selection into treatment \citep{angrist2009mostly, imbens2015causal}, where treatment is not randomized, but unobservable criteria affect both the uptake of the treatment as well as the outcome. For instance, motivation to find a new job after losing it recently will make individuals more likely to sign up for potential job training programs (the treatment) but also makes them more likely to send more applications and land a new job (the outcome), something documented by the classic Ashenfelter Dip \citep{heckman1999pre}. Therefore, simply comparing the outcome in both groups---the people in the job training program versus the rest of the population---does \emph{not} provide the causal effect of attending the job training program. One essentially compares two different groups based on characteristics and would falsely attribute the difference exclusively to the treatment effect of attending the job training program. 

Such problems have been at the center of interest in classical econometrics since the 1920's \citep{haavelmo1943statistical,tinbergen1930determination, wright1928tariff}, starting with the analysis of supply and demand models. Fundamentally, in such settings, identification, estimation, and inference of the effect of the intervention on the outcome of interest requires additional assumptions tailored to the specific case. Over the decades, statisticians and econometricians have developed a plethora of such methods, including instrumental variable estimation, matching, difference-in-differences, and synthetic controls. 

Traditionally, the focus in this literature has been on identifying average effects, essentially estimating an expected effect of the intervention on the outcome of interest in the system. While this is often an important measure, modern research in this area has emphasized the importance of accounting for \emph{treatment heterogeneity} \citep{heckman1997making}. 

A classical and influential example is the literature on the effect of minimum wage increases on employment \cite{card1994minimum, neumark2000minimum}. To analyze the effect a minimum wage increase had on the employment growth rate at fast food chains in the state of New Jersey in $1992$, David Card and Alan Krueger \cite{card1994minimum} used a difference-in-differences approach focusing on average effects; they compared the average change in the employment rate at the surveyed New Jersey fast-food restaurants to the average change of employment at fast-food restaurants in the neighboring Pennsylvania, where the minimum wage remained constant. Somewhat counterintuitively, they found that an increase in the minimum wage increased employment. In stark contrast, David Neumark and William Wascher \cite{neumark2000minimum}, using a different dataset and different methodology, found a negative average effect. 

While the debate of the effects of the minimum wage on employment is still ongoing, an interesting contribution \citep{ropponen2011reconciling} attempts to reconcile these two opposite results by considering the entire distribution of the size of fast-food restaurants in each state---using the changes-in-changes estimator \citep{athey2006identification}, which is an extension of the difference-in-differences idea to account for individual heterogeneity, implicitly built using optimal transportation as shown below. The author finds that the positive effect seems to persist for smaller fast-food restaurants while a negative effect can be found for larger fast-food restaurants. This is a striking example of \revision{how accounting for heterogeneity} in the causal framework can uncover important insights into the problem that an average effect cannot.

The goal of this review is to show how optimal transportation can be used to resolve such and similar problems by allowing to capture \emph{heterogeneity} within the respective setting \citep{heckman2001micro, imbens2007nonadditive}. Importantly, optimal transport allows one to do this in two ways, the second being arguably more novel: the first way fundamentally focuses on identifying causal effects for different individuals in the outcome distribution. In contrast to classical approaches that focus on \emph{average} causal effects, this framework explicitly focuses on the heterogeneity of effects, away from the average, often by considering the quantiles in the potential outcome distributions. The second way is based on a ``systems view''. Here, the goal is not to identify the different individual effects within the system, but to analyze the different counterfactual states of the system \emph{as a whole}.

\revision{It is important to clarify what is meant by ``heterogeneity'' in this review, as the term is used differently across literatures. In much of the causal inference literature, treatment effect heterogeneity refers to variation in average treatment effects across subgroups defined by \emph{observed} covariates, the so-called conditional average treatment effects (CATE). The notion of heterogeneity in this review is different, however. Here, heterogeneity refers to the variation in the causal mechanism induced by the \emph{unobservable} $U$: different individuals, indexed by different realizations of $U$, experience different causal effects of the treatment $X$ on the outcome $Y$. The goal is therefore to identify richer features of the causal mechanism $g(x,u)$ (or at least of the distribution of potential outcomes) than what is captured by an average treatment effect alone. Moreover, the ``systems view'' described above does not focus on individual-level effects at all, but instead on how the \emph{entire distribution} of the outcome changes under different counterfactual states.}

This review therefore has several goals. 
\begin{enumerate}
    \item One goal is to illustrate that not just the methods, but actually the models and assumptions in the literature on identifying causal effects with observational data are implicitly based on optimal transportation. Moreover, I aim to provide a common ground to use optimal transport as a paradigm to connect different areas from econometrics, statistics, machine learning, and artificial intelligence that---knowingly or unknowingly---use optimal transportation in the causal argumentation. This seems useful, because in recent years, there has been an increase in articles explicitly using optimal transportation methods to do causal inference, and many of these approaches \citep[e.g.][]{lin2023causal} do not reference existing related results \citep[e.g.][]{athey2006identification, imbens2007nonadditive}. 
    \item This review first and foremost focuses on \emph{identification} in causal inference with observational data. Identification arguments are of paramount importance in causal inference and precede estimation and inference: researchers need to show that they are actually able to obtain the correct effect in the population before focusing on estimating it. This makes causal inference significantly more complicated than prediction, because the observable distributions in general do not provide correct estimates---unobservable confounders bias estimation results, and one needs special tools to circumvent the resulting endogeneity issues, some of which I outline in this review. 
    \item A third goal is to show how optimal transport can be used to enhance, augment, and generalize many of the existing linear regression methods, in particular when it comes to the identification of heterogeneous effects \citep{heckman1997making}, or identifying effects on entire systems \citep{gunsilius2023distributional}. The key is to generalize the assumptions and models from the literature to more general settings.
    \item Finally, I attempt to point out some open problems and potentially new connections in this quickly expanding area. For instance, I show how realizing that the influential control variable approach \citep{imbens2009identification} is based on optimal transportation allows one to straightforwardly generalize it to make it vastly more applicable in practical settings. As another example, I provide a connection between partial identification in instrumental variable models and optimal transport problems on path space.\revision{This approach has the potential to become a fundamental framework for talking about causal identification. In addition, it showcases how optimal transportation can serve as a bridge between causal inference and statistical physics, allowing for potentially fruitful cross-pollination of these two fields (see also the connection between partial identification of causal effects and Bell's inequality in \cite{pearl1995testability} and \cite{gunsilius2021nontestability}.)}
\end{enumerate}
\revision{This overview} is necessarily biased towards my own area of research, but I made every effort to provide a general overview, trying to connect existing results from diverse areas.

\revision{\textbf{A note on notation.} I decided to make each section of this review follow the notational conventions of the source literature it discusses. In Sections 3--4, $X$ denotes the treatment (or endogenous variable), $Y$ the outcome, $U$ the unobservable, and $Z$ the instrument. In Section 5, $t$ indexes time periods and $g$ indexes treatment/control groups. In Section 6 (synthetic controls), $j$ indexes aggregate units and $X_j$ denotes observed covariates in the factor model. In Section 7 (matching), $T\in\{0,1\}$ denotes the treatment indicator and $X\in\mathbb{R}^d$ the observed covariates. I flag these transitions at the beginning of each section.}

\section{A quick refresher on Optimal Transport}
The basic problem of optimal transportation is the following. Suppose we are given two probability distributions $P$ and $Q$ on some underlying sets (supports) $\mathcal{X}$ and $\mathcal{Y}$. The original problem posed by Gaspard Monge \cite{monge1781memoire} is to find an optimal map $T:\mathcal{X}\to\mathcal{Y}$ that preserves the mass between $P$ and $Q$ and minimizes the overall cost of transporting $P$ onto $Q$ measured in terms of the cost function $c:\mathcal{X}\times\mathcal{Y}\to\mathbb{R}^+_0$, that is
\begin{equation}\label{eq:monge}
\min_{\substack{T:\mathcal{X}\to\mathcal{Y}\\T_\#P = Q}}\int c(x,T(x)) dP(x),
\end{equation}
where $T_\#P(A) = P\left(\left\{x\in\mathcal{X}: T(x)\in A\right\}\right)$ is the pushforward measure from $P$ via $T$. This problem need not have a solution as there need not exist a function $T$ that accomplishes this, for instance when $P$ has fewer points in its support than $Q$. The cost function $c(x,y)$ is given and encapsulates all the important information on the problem. In many settings, it is simply a distance function on the underlying space on which the probability measures are supported. One of the arguably most influential contributions to theory was the following convex relaxation of the problem by Leonid Kantorovich, which instead of asking for a map $T$, only requires a coupling between $P$ and $Q$, that is, a joint distribution $\gamma$ on $\mathcal{X}\times\mathcal{Y}$ that solves the linear program
\begin{equation}\label{eq:kant}
\min_{\gamma\in \Gamma(P,Q)} \int_{\mathcal{X}\times\mathcal{Y}} c(x,y) d\gamma(x,y),
\end{equation}
where $\Gamma(P,Q) = \{\gamma\in P(\mathcal{X}\times\mathcal{Y}): \pi_1\gamma=P, \pi_2\gamma=Q\}$ is the set of all couplings of $P$ and $Q$ \citep{kantorovich1942translocation} and $\pi_1$ (respectively: $\pi_2$) denotes the projection onto the first (respectively: second) marginal distribution. \eqref{eq:kant} is the Kantorovich problem. It is a convex relaxation of the Monge problem \eqref{eq:monge} in the sense that if \eqref{eq:monge} has a solution $T$ then the optimal coupling solving \eqref{eq:kant} is supported on the graph of $T$ in general. Both the minimum and the optimizer are of interest, particularly in economics \citep{galichon2018optimal}, because there one is often interested in the optimal allocation and matching of resources. 

In the case where the cost function is the squared Euclidean distance $c(x,y) = |x-y|^2$, $x,y\in\mathbb{R}^d$, the square root of the value function, i.e., the square root of the minimum achieved via \eqref{eq:kant}, can be shown to possess all the properties of a metric. In this case, the value function is denoted by $W_2^2(P,Q)$ and is called the (square of the) $2$-Wasserstein distance. Technically, it would be more accurate to at least call it the ``Monge-Kantorovich-Wasserstein'' distance, but this nomenclature has stuck. Finally, in the case of the squared Euclidean distance as a cost function and if $P$ is absolutely continuous with respect to Lebesgue measure, then Brenier's theorem \citep[e.g.][Theorem 2.12]{villani2021topics} implies that the optimal transport map $T$ in \eqref{eq:monge} takes the form of the gradient of a convex function, i.e., $T(x)=\nabla\varphi(x)$ for some convex $\varphi$. This property has been a main contributor to a surge of optimal transport methods in statistics and econometrics \citep[e.g.][]{carlier2016vector, del2024nonparametric, gunsilius2023independent}, because the gradient of a convex function is a natural generalization of a monotone function. 

Since then the contributions to theory have exploded to dynamic \citep{benamou2000computational}, weak \citep{backhoff2022applications, gozlan2017kantorovich}, multimarginal \citep{agueh2011barycenters, carlier2010matching}, regularized \citep{cuturi2013sinkhorn}, unbalanced \citep{chizat2018unbalanced, liero2018optimal, sejourne2023unbalanced}, and geometric representations \citep{jordan1998variational}, just to name a few. In the following, I introduce the respective mathematical theory whenever it is needed. For general overviews of the basic concepts, consider \cite{peyre2019computational, rachev2006mass1, santambrogio2015optimal, villani2008optimal, villani2021topics}. 

Fundamentally, the fact that optimal transportation induces a coupling between probability measures by \emph{mapping between the underlying state spaces} $\mathcal{X}$ and $\mathcal{Y}$ provides a natural connection to causal inference if one considers $P$ and $Q$ probability measures of outcomes of interest and unobservable confounders. To showcase the utility of optimal transportation, I now introduce a classical setting of structural equation models and show how optimal transportation has implicitly provided the underlying structure for the fundamental identification and estimation approaches in this area.

\section{The foundation: structural models, counterfactuals, and the monotone rearrangement}
This section analyzes how optimal transportation allows for the identification of causal effects while accounting for individual heterogeneity. In fact, I show how ideas from optimal transportation, in particular the \emph{monotone rearrangement}, have built the foundation for models and assumptions in this area and how this insight can be used for generalizations and to develop new methods.
\revision{In this section and the next, $X$ denotes the treatment (or endogenous variable), $Y$ the outcome, and $U$ the unobservable.}

\subsection{Structural models}
A way to formalize the setup from the introduction is via \emph{structural models} of the form
\begin{equation}\label{eq:struct_mod}
Y= g(X,U),
\end{equation}
where $Y:\Omega \to \mathbb{R}$ is the observed outcome of interest and $X:\Omega\to \mathbb{R}^d$ is the observed treatment variable, that is, the variable in the system whose effect on the outcome we want to isolate. $U:\Omega\to \mathcal{U}$ accounts for all other unobserved influences in the system, in particular the \emph{treatment heterogeneity}: different values of $U$ introduce different causal mechanisms between $X$ and $Y$. $\mathcal{U}$ can potentially be infinite dimensional. All variables are defined on a common probability space $(\Omega, \mathcal{A}, P)$. In many classical causal inference settings $X$ is a univariate binary variable taking values in $\{0,1\}$, indicating the treatment status, \revision{but it can be more general}. The unobserved parts in this model are the variable $U$ and the causal mechanism $g(X,U)$. 

The generality in the assumption of the causal mechanism serves to not introduce additional structural assumptions into the model besides a postulated relationship between the variables $Y$, $X$, as part of the larger system. $U$ captures all the unobservable factors in the system that affect the outcome $Y$ and potentially the treatment $X$. This means that $U$ and $X$ are \emph{not} assumed to be independent, as would be the case in a randomized controlled trial. This captures the fundamental problem of causal inference in this setting, \revision{as the dependence between $X$ and $U$ opens a back-door path \citep{pearl2009causality} that confounds the identification of the causal effect $X\to Y$} in Figure \ref{fig:DAG_structure}. 

This induced ``feedback'' loop is the classical \emph{endogeneity problem}, i.e., $X$ is endogenous to the system so that the effect $X\to Y$ cannot be extracted directly by simple prediction. Therefore, for identification in these settings one needs more information and stronger assumptions to close the back-door path $U\to X$, as I will discuss below in examples such as instrumental variable models and difference-in-differences estimation. Before doing this, we first need to introduce the foundational ingredients for identifying causal effects in this setting by outlining how identification of treatment effects has been based on the \emph{monotone rearrangement}, a solution to specific types of optimal transport problems.
\begin{figure}[h!]
    \centering

\begin{tikzpicture}[
    node distance=1cm, 
    box/.style={rectangle, draw, minimum width=0.8cm, minimum height=0.6cm}, 
    circle/.style={ellipse, draw, minimum width=0.8cm, minimum height=0.6cm}, 
    arrow/.style={->, >=Stealth} 
]

\node[box] (X) {\( X \)};
\node[box, right of=X, xshift=2cm] (Y) {\( Y \)}; 
\node[circle, above of=X, yshift=0.1cm, xshift=1.3cm] (U) {\( U \)}; 

\draw[arrow] (U) -- (X);
\draw[arrow] (U) -- (Y);
\draw[arrow] (X) -- (Y);

\end{tikzpicture}
    \caption{The DAG corresponding to \eqref{eq:struct_mod} illustrating the backdoor path through the unobservable $U$.}
    \label{fig:DAG_structure}

\end{figure}
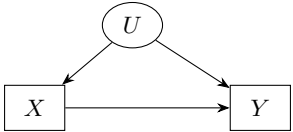

The model \eqref{eq:struct_mod} is deliberately general and in particular encompasses additively separable structures ($g(X,U) = f(X)+U$) and even more restrictive linear models ($g(X,U) = X^\top\beta + U$). The focus is also on identifying the mechanism $g(X,u)$ directly, i.e., the relationship $X\to Y$ while controlling for the unobservable $U=u$. The heterogeneity is captured by the unobservable $U$, which contains all characteristics that are unobserved to the statistician, but are of importance to the causal question of interest. Allowing for $g$ to depend nonseparably on $U$ allows for this heterogeneity to be of the most general form, without the need to introduce artificial structural assumptions like additive separability. It has been argued that additive separability is an artificial assumption in most settings of interest when dealing with human behavior \citep[e.g.][]{hoderlein2007identification}.

\subsection{Relation to counterfactuals and other causal questions}
Structural models also induce the classical counterfactuals  $Y(x)$ \citep{neyman1990application, rubin1974estimating, holland1988causal}. Here, $Y(x)$ is the counterfactual outcome if the treatment $X$ takes the value $x$. In the classical binary setting $X\in\{0,1\}$, $Y(1)$ is the counterfactual under treatment while $Y(0)$ is the counterfactual under no treatment. The fundamental problem of causal inference  for an independent and identically distributed sample $\{Y_i,X_i\}_{i=1}^n$ in the simple binary case $X\in \{0,1\}$ can be stated simply as 
\begin{equation}\label{eq:rubin}
Y_i = (1-X_i) Y_i(0) + X_i Y_i(1), 
\end{equation} that is, one can only ever observe one potential outcome in practice. Structural equations are sometimes argued to be more general than the counterfactual notation because they explicitly note the importance of the unobserved heterogeneity $U$ in the model \citep{pearl2009survey}. In particular, equivalent to structual equation models are the \emph{counterfactual processes} $Y_x(u)$, which denote the counterfactual outcome under treatment $x$ if the system is in the state $U=u$ \citep{balke1994counterfactual, pearl2009causality}. Below I show how those processes induce new and interesting questions that relate to optimal transportation---analyzing those can open connections between optimal transport, statistical physics, and causal inference \revision{\citep{gunsilius2025schrodinger}}.

The model \eqref{eq:struct_mod} is completely general in that it does not impose any structural assumptions on the system, and in particular $U$, \emph{a priori}. It is also different in focus and in a sense more general than the models in other areas of causal inference for machine learning \cite[e.g.][]{peters2017elements}. The difference is that the goal in our setting is not to analyze the system and to understand the relationships between all the variables in the system relegated to the error term $U$, but to simply isolate the causal effect $X\to Y$ among all the other moving parts (captured by $U$) in the system. Intuitively, for our questions of causal inference with observational data, we simply want to \emph{control} the unobservable variable $U$, while in other settings it might be interesting to explicitly model the relationships among all the unobservable variables and their effects on the outcome via directed acyclic graphs  \citep{pearl2009causality, peters2017elements}. Optimal transport is also being used in such settings \citep{cheridito2023optimal, tu2022optimal}, based on the idea of optimal transport between different processes \citep{backhoff2022stability}, exploiting the Markovian structure of DAGs. 

The difference between those two viewpoints is important from an applied re\-search\-er's perspective, as one often wants to be as agnostic as possible about the other relationships in the system and only focus on the one relationship between $X$ and $Y$. Consider \cite{imbens2020essay} for further reading about this distinction. Throughout, this review concentrates on the latter, i.e., extracting the effect $X\to Y$, not modeling the latent space. I now show how optimal transport has built the foundation for this, via the monotone rearrangement.

\subsection{Setup: identification and monotone rearrangement}
The key question is to identify, that is, isolate, the causal effect $X\to Y$ from the system. ``Identification'' is defined as injectivity of the observable law $P_{Y,X}$ induced by the model \eqref{eq:struct_mod} in unobservable quantities $g, P_U$, where the latter is the law of $U$ \citep{matzkin2003nonparametric, pearl2009causality}. The focus for identification is the function $g$, as it encodes the effect $X\to Y$. \revision{Without additional assumptions}, it is impossible to identify the mechanism $g$; in such cases one can only obtain bounds on specific functionals, as I show below. I now quickly outline the predominant assumption to identification in this setting and show how it is fundamentally based on optimal transportation.

Let $G$ be the set of all functions $g$ we consider and $\Gamma$ a set of laws $P_U$ of $U$ in question. Formally, the definition of identification is the following.
\begin{definition}[\cite{matzkin2003nonparametric}]\label{def:matzkin}
    The pair $(g, P_U)$ is identified in the set $(G \times \Gamma)$ if (i) $(g, P_{U}) \in (G \times \Gamma)$, and (ii) for all $(g', P'_{U})$, in $(G \times \Gamma)$,
\begin{multline*}
[P_{Y,X}(\cdot,\cdot; g, P_{U}) = P_{Y,X}(\cdot,\cdot; g', P'_{U})] \\\Rightarrow (g, P_{U}) = (g', P'_{U}).
\end{multline*}
\end{definition}
This means that the map from the observable joint distribution $P_{Y,X}$ of outcome and treatment is injective in the model: two different models generate different observable outcomes. This would lead to point-identification, i.e., obtaining a unique function $g(x,u)$.

\subsubsection{The basic idea in exogenous settings}\mbox{}\\
Before stating approaches for identification when $X$ is truly endogenous in the model, I focus on the case where $X$ is actually exogenous, that is, where $X$ and $U$ are independent (denoted in the following by $X\perp U$), so that the arrow $U\to X$ in Figure \ref{fig:DAG_structure} is absent. \revision{This setting is significantly simpler}, actually circumventing the main causal inference problem, but it serves to illustrate the first fundamental connection between optimal transport and such nonseparable representations. 

A simple sufficient condition for identification in the case where all variables $Y, X$, and $U$ are univariate, $F_U$, the cumulative distribution function (CDF) of the law of $U$ is absolutely continuous and $X\perp U$ is monotonicity of the causal mechanism. In fact, the main assumption that guarantees identification of the causal mechanism up to some equivalence class is \emph{continuity and strict monotonicity} of $g$ in $U$ for all $x$ \citep{doksum1974empirical}. This follows from the simple string of equalities for almost every $x$:
\begin{equation}\label{eq:matzkin_argument}
\begin{aligned}
    F_{Y|X=x}(y) = &P(Y\leq y|X=x) \\
    = & P(g(X,U)\leq y|X=x)\\
    =& P(U \leq g^{-1}(X,y)|X=x) \\
    =& F_{U}(g^{-1}(x,y)),
\end{aligned}
\end{equation}
where $F_{Y|X}$ is the conditional CDF of the observables $Y$ and $ X$. The invertibility of $g$ follows because it is assumed to be continuous and strictly increasing. 

This simple argument in the univariate case implies that the function $g$ is identified up to observational equivalence in the exogenous setting $X\perp U$; in this sense two functions $g,g'\in G$ are observationally equivalent if there exist $F_U, F'_U\in \Gamma$ such that they all induce the same observational distribution, i.e.,
\[F_{Y,X}(\cdot, \cdot; g, F_U) = F_{Y,X}(\cdot, \cdot; g', F'_U).\]
This implies that the entire mechanism $g$ is identified up to a monotone transformation, since any monotone transformation will provide an observationally equivalent mechanism. 

While this argument seems straightforward, its impact on the literature on identification of effects cannot be overstated. In fact, this idea is used in many fundamental approaches to identify causal effects, such as instrumental variable models \cite{imbens2009identification, heckman2005structural, torgovitsky2015identification} and extensions of the difference-in-difference method \citep{athey2006identification}, which I revisit below. 
It is thus interesting to analyze the limits and potential extensions of this idea. This is facilitated by the connection to optimal transportation.

\subsubsection{The monotone rearrangement} \mbox{}\\
The connection to optimal transport comes from the fact that the identification argument \eqref{eq:matzkin_argument} implies that $g(x,\cdot)$ is the \emph{monotone rearrangement} between the unobservable $F_U$ and the observable $F_{Y|X=x}$ for $P_X$-almost every $x$. That is, $g$ is the solution to the Monge problem \eqref{eq:monge} transporting $P_U$ to $P_{Y|X=x}$ for any symmetric convex cost function $c(|u-y|)$, which exists if $P_U$ does not have atoms (that is, gives positive mass to single points) \citep{villani2021topics}. 

The monotone rearrangement $g(x,\cdot)$ takes the form
\begin{equation}\label{eq;monotone_rearrangement}
g(x,u) = F^{-1}_{Y|X=x}\left(F_U(u)\right),
\end{equation}
where 
\(F^{-1}_Z(q) = \inf\left\{z\in\mathbb{R}: F_Z(z)\geq q\right\} \) is the quantile function of the random variable $Z$. It is an optimal transport map in the sense of Monge \eqref{eq:monge}, because (for fixed $x$) it maps every point $u$ in the support of $P_U$ to a point $y= T(u) = g(x,u)$ in the support of $P_{Y|X=x}$. Moreover, it is \emph{measure preserving}, meaning that the preimage of a measurable subset $B$ in the support of $Y|X=x$ gets mapped to a subset $A = T^{-1}(B)$ in the support of $U$ that has the same measure. Importantly, the monotonicity requirement makes this map the unique measure- and order preserving transformation, as depicted in Figure \ref{fig:monotone_rearrangement}. \revision{Note that measure-preservation is a consequence of the structural model $Y=g(X,U)$; any valid $g(x,\cdot)$ must push $P_U$ forward to $P_{Y|X=x}$. There are infinitely many such maps in general. It is the additional order-preservation of the monotone rearrangement that pins down $g$ uniquely and is thus the identifying assumption.}
\begin{figure}[h]
    \centering
    \begin{tikzpicture}[scale=0.7]

        \draw[->, thick] (0,0) -- (3.5,0) node[right] {$U$};
        \draw[->, thick] (0,0) -- (0,3.3) node[above] {$F_U$};
        
        \draw[thick] (0,0) to[out=30,in=190] (3,2.8);
        
        \draw[->,dashed, thick] (1.45,0) -- (1.45,0.7);
        \draw[dashed, thick] (1.45,0) -- (1.45,1.6);
        \draw[->,dashed, thick] (0,1.6) -- (4,1.6);
        \draw[dashed, thick] (4,1.6) -- (8.2,1.6);
        \draw[dotted, thick] (0,2.8) -- (8.2,2.8);

        \node[below] at (1.5,0) {$u_0$};
        \node[left] at (0,1.6) {$F_U(u_0)$};
        \node[left] at (0,2.8) {$1$};

        \draw[->, thick] (5,0) -- (8.5,0) node[right] {$Y | X = x$};
        \draw[->, thick] (5,0) -- (5,3.3) node[above] {$F_{Y|X=x}$};
        
        \draw[thick] (5,0) to[out=20,in=240] (8,2.8);

        \draw[->,dashed, thick] (7.2,1.6) -- (7.2,0.7);
        \draw[dashed, thick] (7.2,0.7) -- (7.2,0);
        \node[below] at (7.2,0) {$y$};
    \end{tikzpicture}
    \caption{Depiction of the monotone rearrangement $y=g(x,u_0)$$ = F^{-1}_{Y|X=x}(F_U(u_0))$}
    \label{fig:monotone_rearrangement}
\end{figure}
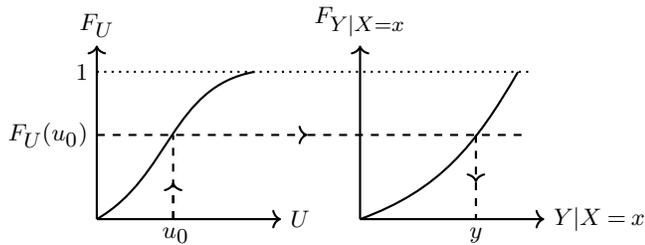

In this case, the monotonicity assumption for identification can be rephrased as: suppose $U$ is a univariate index accounting for the heterogeneity of the system in question and assume that the causal mechanism captured by $g(x,u)$ is such that for each hypothetical state $u$ of the system, the counterfactual outcome $Y(x)$ is created as the optimal transport coupling induced by a convex cost function. Then if the treatment is exogenous (i.e., $X\perp U$), the causal mechanism is the solution to the optimal transport problem with convex symmetric cost. 

\revision{It is worth spelling out what identifying the entire causal mechanism $g(x,u)$ achieves beyond population-level quantities such as the average treatment effect $\mathbb{E}[Y(1) - Y(0)]$. Knowledge of $g$ yields the full distribution of counterfactual outcomes $Y(x) = g(x, U)$ for each treatment level $x$, and thereby all distributional features: quantile treatment effects, the variance of effects, or the probability that the treatment harms an individual. The average treatment effect can be zero while treatment benefits some individuals and harms others; the mechanism $g$ captures this, while the ATE alone cannot. Moreover, the monotone rearrangement specifies which quantile of $Y(0)$ maps to which quantile of $Y(1)$, providing a coupling between counterfactual outcomes that provides a joint distribution of the potential outcomes. Finally, in the endogenous settings discussed in Sections 4--6, assumptions such as monotonicity are not just needed to identify $g$, but also for identifying even average effects.}

This simple connection \revision{allows one to} analyze and extend \eqref{eq:matzkin_argument} in many different ways. For one, it connects this argument to other areas, in particular to \emph{structural nested distribution models} \citep{ robins1989analysis, robins1992estimation, vansteelandt2014structural}, as pointed out in \cite{balakrishnan2023conservative}. Structural nested distribution models are \revision{generalizations} of structural nested mean models, where the relationship between counterfactual outcomes \revision{$Y(0)$ and $Y(x)$} is modeled by \revision{a parametric function} 
\[\revision{\alpha(y,x,l,\beta) = F^{-1}_{Y(0)|L=l,X=x}\left(F_{Y(x)|L=l, X=x}(y)\right),}\] \revision{known up to the parameter} $\beta$ \revision{and conditional on observed confounders $L=l$}. \revision{In short, for structural nested distribution models \citep{laan2003unified}, $\alpha$ is a parametric model for the optimal transport map between the counterfactual distributions, indexed by $\beta$. The parameter $\beta$ is identified from the observed data through g-estimation \citep{vansteelandt2014structural}: one defines the transformed outcome $U(\beta) = \alpha(Y, X, L, \beta)$, which ``undoes'' the treatment effect. Under ignorability, the assumption that $Y(x) \perp X \mid L$ for all $x$, the true $\beta$ is the value that makes $U(\beta)$ independent of treatment $X$ conditional on $L$ (a condition that is testable from the observed data).}

Moreover, the connection of \eqref{eq:matzkin_argument} to the monotone rearrangement links the identification argument to the classical Fr\'echet-Hoeffding bounds \citep{rachev2006mass1}. The coupling $\gamma$ induced by the monotone rearrangement $g(x,u)$, which solves the Kantorovich problem \eqref{eq:kant} has the CDF 
\[H(u,y;x) = \min\{F_U(u), F_{Y|X=x}(y)\}\] for $P_X$-almost every $x$. This is the upper Fr\'echet-Hoeffding copula, the copula that makes $F_U$ and $F_{Y|X=x}$ \emph{comonotone}, maximizing their dependence. 

This rephrasing clarifies the obvious restrictiveness of the monotonicity assumption in \eqref{eq:matzkin_argument}. In many settings, comonotonicity is reasonable. For instance when analyzing a hypothetical setting of the causal effect of the grade-point average in high-school ($X$) on future earnings ($Y$) of students, where the unobservable $U$ is understood as an ``index of ability'' of the student. In this case, it is reasonable to assume that students with higher ability will also earn more conditional on their GPA. 

\revision{However, it is important to note that the monotonicity assumption comes with strong restrictions: the resulting map is not only measure-preserving but also order-preserving. In a dynamic setting, for instance, this would imply that an individual at the $10^{th}$ percentile of the unobservable can never change their rank. More broadly, monotonicity is just one point in a larger space of assumptions on the causal mechanism $g$, indexed by the choice of cost function in the optimal transport problem. In one dimension, cost functions of the form $c(u-y)$ with $c$ convex, non-negative, and symmetric, yield the monotone rearrangement \citep[Theorem 2.18]{villani2021topics}. }

\revision{In contrast, concave cost functions lead to fundamentally different maps $g$ that are in general not monotone: as shown in \cite{mccann1999exact}, the resulting optimal maps exhibit a hierarchical structure that mixes local and long-range transport. Such mechanisms might be better suited for settings where extreme realizations of $U$ lead to qualitatively different outcomes. For instance, settings where very high or very low values of an unobservable index produce similar adverse effects. At present, monotonicity remains the dominant assumption largely due to its simplicity and interpretability. But realizing that these identification strategies are optimal transport problems in disguise opens the door to a richer family of assumptions, and potentially to sensitivity analyses that vary the cost function to assess the robustness of the identified causal mechanism.}

\subsubsection{Connection to the counterfactual notation}\mbox{}\\
\revision{The connection to the Fr\'echet-Hoeffding copulas also allows for a connection to Rubin's counterfactual model. Consider the binary case $X\in \{0,1\}$ for simplicity and maintain the assumption that $X\perp U$. By \cite{dawid1979conditional}, this implies $(Y(1), Y(0)) \perp X$, so that the marginal distributions of $Y(1)$ and $Y(0)$ are identified from the observed data by considering the two groups separately. The important question, however, is to identify the \emph{joint} distribution of the potential outcomes $(Y(0),Y(1))$, which captures the causal mechanism. Without further assumptions, this joint distribution cannot be identified, since by \eqref{eq:rubin} one only ever observes one potential outcome per unit. This is precisely a question of which coupling between $F_{Y(0)}$ and $F_{Y(1)}$ is the correct one.}

One can  obtain bounds on the causal mechanism of interest by the standard Fr\'echet-Hoeffding bounds. In the binary setting, this has been done in \cite{aronow2014sharp, fan2010sharp, heckman1995assessing}. The Fr\'echet-Hoeffding bounds for the joint distribution of $(Y(0),Y(1))$ are \citep[e.g.][Theorem 3.1.1]{rachev2006mass1}
\begin{multline}\label{eq:frechet}
F_{Y(0)}(y) + F_{Y(1)}(y')-1 \\ \leq F_{Y(0),Y(1)}(y,y')\\\leq \min \left\{F_{Y(0)}(y),F_{Y(1)}(y')\right\}. 
\end{multline}
Based on this, one can bound other expressions such as the covariance $\mathrm{Cov}(Y(0),Y(1))$ or any expression of the form $\mathbb{E}[c(Y(0),Y(1))]$ via
\begin{multline}
    \int_0^1 c(F^{-1}_{Y(0)}(u), F^{-1}_{Y(1)}(u))\, du\\ \leq \mathbb{E}[c(Y(0),Y(1)]\\\leq \int_0^1 c(F^{-1}_{Y(0)}(u), F^{-1}_{Y(1)}(1-u))\, du
\end{multline}
for any quasi-antitone cost function $c$, that is, 
\[c(x',y')+c(x,y)\leq c(x',y)+c(x,y')\] for any $x'\geq x$ and $y'\geq y$, which is a classic inequality proved in \cite{cambanis1976inequalities}. 
These bounds are wide in general, often too wide to be of use, but they do provide a minimal restriction on the respective data-generating process. 

I would be \revision{remiss} to not mention that recently, \cite{balakrishnan2023conservative} analyzed lower bounds on functionals of the joint distribution of all potential outcomes $Y(x)$ in a setting where the treatment $X$ is continuous. By solving a general optimal transport problem, i.e., minimizing functionals such as the ``quadratic effect''
\[
\mathbb{E} \left[ \int \int |Y(x) - Y(x')|^2 \, dP_X(x) \, dP_X(x') \right],
\]
where $P_X$ is the marginal distribution of the treatment $X$ and the expectation is over all joint laws of the counterfactual outcomes $\{Y(x)\}_{x\in \mathbb{R}}$ they obtain lower bounds on these effects. It can hence be seen as a generalization of the Fr\'echet-Hoeffding approach to continuous treatments. \revision{The reasoning is that the set of all joint laws consistent with the identified marginals contains the true (unknown) joint law; minimizing over this set therefore yields a value no larger than the true causal effect, i.e., a sharp lower bound.}

The same issue as above holds here: the bounds are usually wide without additional assumptions. Moreover, it is not always clear when a quadratic effect like this is interesting---recall that the assumption of the monotone rearrangement as the causal mechanism allows one to identify the entire mechanism, not just functionals of it. Below in section \ref{sec:IV_bounds} I make a connection to the problem of bounding treatment effects in instrumental variable models \citep{balke1994counterfactual, balke1997bounds}, where more information is introduced into the above problem via linear constraints, leading to tighter bounds and a novel optimal transport adjacent optimization problem on path spaces. 

\subsubsection{Monotone rearrangement as the foundation of causal inference with observational data}\mbox{}\\
The monotone rearrangement has been the workhorse for the identification of causal effects in a variety of different areas. Most likely, the main reason for this is its interpretability, it being an order preserving map. We have also seen, however, that it is in general \emph{not} enough to identify the causal effects without stronger assumptions. First, in the structural setting with $Y=g(X,U)$, we had to assume that $X\perp U$, which is the simple setting of a randomized controlled trial, circumventing the causal identification problem entirely---and then we could only identify the mechanism up to a monotone transformation. \revision{Second, even in the exogenous setting, from the perspective of the counterfactual notation, the monotone rearrangement only provides bounds on the joint distribution of the potential outcomes in general.} On the other hand, the monotone rearrangement is the fundamental connection to works on the intersection of optimal transport and causal inference \cite{de2024transport} that argue that the fundamental model between potential outcomes should be induced by optimal transport: in a way, this has always been the case, but as a connection between unobservables and outcome, not between the counterfactual distributions directly. 

In the following, I therefore have two goals. First, I want to show how the realization that the monotone rearrangement is the solution to certain one-dimensional optimal transport problems can be used to generalize causal inference to higher dimensions and more realistic settings. Second, I want to introduce and generalize classical approaches such as instrumental variable estimation, difference-in-differences, and synthetic controls to \emph{uniquely identify} either the entire causal mechanism $g(x,u)$---or causal effects based on it---in the case where $X$ is endogenous, i.e., $X\not\perp U$, by exploiting additional structure and assumptions in the problem.

\subsubsection{Problems with extension to more general settings}
Realizing the connection between \eqref{eq:matzkin_argument} and optimal transportation, there is now a straightforward way to generalize the setting to higher dimensions. That is, instead of requiring that $U$ and $Y$ be univariate and $g(x,u)$ be the monotone rearrangement between $U$ and $Y|X=x$, we can assume that $U$ and $Y$ are multivariate \emph{of the same dimension} and that $g(x,u)$ is the optimal transport map for the quadratic cost function $c(u,y) = |u-y|^2$. Considering $x$ fixed, by Brenier's theorem \citep[Theorem 2.12]{brenier1991polar, villani2021topics}, this map is uniquely defined $P_U$-almost everywhere and takes the form of the gradient of a convex function, that is $g(x,u) \equiv T_x(u) = \nabla\varphi_x(u)$ for some convex function $\varphi_x:\mathbb{R}^d\to\mathbb{R}$. Identification would then come from the fact that the Brenier map is the unique mapping with this monotonicity property, analogous to the monotone rearrangement. 

Brenier's theorem has been the starting point for a recent explosion in interest in applications of static optimal transportation in statistics \citep[e.g.][]{carlier2016vector, chernozhukov2017monge, chernozhukov2021identification, del2024nonparametric, fan2022lorenz}. The reason being that the gradient of a convex function is a ``natural'' generalization of a monotone function. In fact, it is not only monotone as a map $\mathbb{R}^d\to\mathbb{R}^d$, in the sense that 
\[\langle T_x(u) - T_x(u'), u-u'\rangle\geq 0\quad\text{for all $u,u'\in \mathbb{R}^d$},\] where \(\langle \cdot, \cdot \rangle\) denotes the inner product on \(\mathbb{R}^d\times\mathbb{R}^d\), but it actually possesses a \emph{cyclically monotone} support \citep[Theorem 24.8]{rockafellar1970convex}. 
A map \( T_x : \mathbb{R}^d \to \mathbb{R}^d \) is said to be cyclically monotone if for any positive integer \( m \) and any cycle \( u_1, \ldots, u_m, u_{m+1} \equiv u_1 \) in its domain, it holds

\[
\sum_{i=1}^{m} \langle u_i, T_x(u_i) - T_x(u_{i+1}) \rangle \geq 0.
\]
In the univariate setting, monotonicty and cyclic monotonicity coincide, but cyclic monotonicity provides more structure in multivariate settings, as the inequality has to hold for all finite $m$-cycles, while monotonicity only requires this for $m=2$. \revision{To illustrate: any map of the form $T_x(u) = Au$ with $A$ symmetric positive semidefinite is cyclically monotone, since it is the gradient of the convex function $\varphi(u) = u^\top A u/2$. More generally, the Brenier map between two multivariate Gaussians $\mathcal{N}(\mu_0, \Sigma_0)$ and $\mathcal{N}(\mu_1, \Sigma_1)$ is cyclically monotone: it is an affine map $T(u) = \mu_1 + A(u - \mu_0)$ where $A = \Sigma_0^{-1/2}(\Sigma_0^{1/2}\Sigma_1\Sigma_0^{1/2})^{1/2}\Sigma_0^{-1/2}$, which is again the gradient of a convex function. In contrast, a monotone but \emph{not} cyclically monotone map---such as a rotation---would not arise as an optimal transport map for the quadratic cost. Cyclic monotonicity thus restricts the class of admissible mechanisms $g$ more tightly than monotonicity alone, and it is exactly this additional structure that yields uniqueness of the Brenier map and hence identification.} Below, I show how cyclic monotonicity is useful in the identification of individual effects. 

While this extension is straightforward and useful in many areas of statistics, it is less useful in our setting. The main issue is the interpretation of the requirement that $U$ and $Y$ have to be of the same dimension. First, in causal inference settings, it is less common to have multivariate outcomes. Second, the interpretation of the unobservable in  a multivariate setting is difficult. $U$ is designed to contain any covariate that can theoretically affect the outcome of interest and is not observable. Therefore, either assuming $U$ is infinite dimensional or that $U$ is univariate seem reasonable. In the first case, it is a general and all-encompassing error term. In the second case it is an index that contains the information of all relevant unobservables (for instance an ``ability''-index in the GPA example). The choice that $U$ needs to be multivariate and of the same dimension as $Y$ is therefore artificial in most applications.

In this sense, it is more realistic to assume $Y$ is univariate or low-dimensional while $U$ is higher-dimensional. Such generalizations have been analyzed in \cite{chiappori2017multi, mccann2020optimal}. The issue here is that unlike the case where the dimensions of $Y$ and $U$ align, the existence (and uniqueness) of an optimal transport map $g(x,u)$ between $u\in\mathbb{R}^d$ and $y\in \mathbb{R}$ relies on the shapes of $P_U$ and $P_Y$ as well as the cost function $c(u,y)$ in the optimal transport problem \citep[Theorem 4.(b)]{chiappori2017multi}. When such a map exists in the setting $Y\in\mathbb{R}$, the authors call the corresponding model \emph{nested}. The fact that the existence of $g(x,u)$ depends on the geometries of $P_{Y|X=x}$ and $P_U$ is interesting in general causal inference problems. It suggests that there are important differences in models that allow to identify causal models in these settings. Also, generalizing the results in \cite{chiappori2017multi, mccann2020optimal} to allow for infinite dimensional $U$ is interesting. 

While the standard generalization of the monotone rearrangements to Brenier maps in multiple dimensions is often not a realistic model for the relation between the outcome $Y$ and the unobservables $U$ in the system, it has proven to be very helpful and illuminating in actual approaches to causal inference when $X\not\perp U$; and there multivariate extensions are in fact important, as they allow to consider multiple endogenous treatments. I now review three of these approaches. The first is the \emph{method of instrumental variables}, the second is difference-in-differences, and the third is synthetic controls. 

\section{Instrumental variable models: identification, bounds, robustness}
\revision{This section retains $X$, $Y$, $U$ from above and introduces $Z$ for the instrument, $W$ for the first-stage unobservable, and $V$ for the second-stage unobservable.}

\revision{This} is the first instance where we start dealing with the true causal inference problem depicted in Figure \ref{fig:DAG_structure}, i.e., where $X\not\perp U$. As indicated above, it is not possible to point-identify the structure of $g(x,u)$ without additional assumptions, \emph{even if $g$ is assumed to be the monotone rearrangement}. The reason is the dependence between $X$ and $U$, the backdoor channel (Figure \ref{fig:DAG_structure}). It implies that the observed conditional measure $P_{Y|X=x}$ is not the counterfactual measure of interest, which is denoted by $P_{Y|X^*}$, indicating that this is the conditional measure for the hypothetical scenario where $X\sim X^*$ but $X^*\perp U$. The observed joint measure $P_{Y,X}$ is expressed as
\[
    P_{Y, X} = \int P_{Y|X, U=u} P_{X|U = u} P_U(du).
\]
If $X$ were actually exogenous, that is, $X\perp U$, we would have $P_{X|U} = P_{X}$, so that we could identify $P_{Y|X^*}$ by

\begin{align*}
    P_{Y|X} =& \frac{P_{Y, X}}{P_{X}} 
    = \frac{P_{X} \int P_{Y|X, U = u} P_U(du)}{P_{X}} = P_{Y|X^*}.
\end{align*}
Since $X^*$ is unobservable, one idea is to find a random variable $Z$ that has its properties. In short, we would want a variable $Z$ that has the same distribution as $X$ but is itself independent of the unobservable $U$. This is too strong a restriction, so we require $Z$ to be an \emph{instrumental variable} with the following properties:
\begin{definition}
    A random variable $Z:\Omega\to\mathbb{R}$ is an instrument for the endogenous variable $X$ in the model \eqref{eq:struct_mod} if 
    \begin{enumerate}
        \item[(i)] (Relevance) $Z\not\perp X$,
        \item[(ii)] (Independence) $Z\perp U$,
        \item[(iii)] (Exclusion) The only influence $Z$ has on the outcome $Y$ is via $X$. 
    \end{enumerate}
\end{definition}
If $Z$ satisfies (i) - (iii), we call it \emph{valid} \citep{imbens2015causal, pearl2009causality}. Incorporating an instrument into our model, we can extend \eqref{eq:struct_mod} to the following structural model.
\begin{equation}\label{eq:struct_IV}
    \begin{aligned}
        Y &= g(X,V)\\
        X &= h(Z,W), \quad Z\perp (W,V),
    \end{aligned}
\end{equation}
where we have split the unobservable error term $U$ into an error term $W$ of the \emph{first stage} and an error term $V$ of the second stage. This model is equivalent to a model where $U$ appears in both the first- and second-stage, but is more convenient for the derivations below.

This model implicitly contains all the information we require from an instrumental variable model. Relevance is given if $h$ is not constant in $Z$ for every $u$. Independence is enforced by $Z\perp U$, and exclusion is captured by the fact that $g$ is not a function of $Z$. One can depict \eqref{eq:struct_IV} more elegantly as a DAG \citep{pearl2009causality} as in Figure \ref{fig:DAG_IV}.
\begin{figure}[h!]
    \centering
    \begin{tikzpicture}[
        node distance=2cm, 
        box/.style={rectangle, draw, minimum width=0.8cm, minimum height=0.6cm}, 
        circle/.style={ellipse, draw, minimum width=0.8cm, minimum height=0.6cm}, 
        arrow/.style={->, >=Stealth} 
    ]
    
    \node[box, xshift=-1cm] (Z) {\( Z \)}; 
    \node[box, right of=Z, xshift=1cm] (X) {\( X \)};
    \node[box, right of=X, xshift=1.5cm] (Y) {\( Y \)};
    
    \node[circle, above of=X, xshift=0.4cm] (W) {\( W \)};
    \node[circle, above of=X, xshift=2cm] (V) {\( V \)};
    
    \node[draw, dashed, ellipse, fit=(W) (V), inner sep=0.4cm] (U) {};
    \node[above of=U, yshift=-1.5cm] {\( U \)};
    
    \draw[arrow] (Z) -- (X);
    \draw[arrow] (X) -- (Y);
    \draw[arrow] (W) -- (X);
    \draw[arrow] (V) -- (Y);
    \draw[arrow] (W) -- (V);
    
    \end{tikzpicture}
    \caption{The DAG corresponding to \eqref{eq:struct_IV}.}
    \label{fig:DAG_IV}

\end{figure}
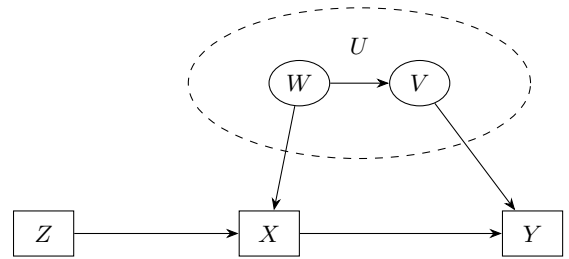

To understand the intuitive idea, in a simple linear setting, i.e., where $g(X,V) = \alpha X+ V$ and $h(Z,W) = \beta Z+ W$, one can identify the average causal effect $ X\to Y$ by the coefficient $\beta^{IV}$, which can be obtained via the Wald estimator $\beta^{IV} = \frac{\mathrm{CoV}(Y,Z)}{\mathrm{CoV}(X,Z)}$. This estimator provides an intuitive explanation in how instruments can be used to identify the causal effect even in more general settings. Since $Z$ does not have a direct influence on $Y$ and is itself independent of $U$, any correlation between $Y$ and $Z$ must go through (or in causal inference speak: ``is mediated by'') $X$. Therefore, to identify the causal effect in this setting, we only have to normalize this relationship by the relation between $X$ and $Z$ to identify the (linear) average causal effect $X \to Y$. 

While similar, this argument is significantly more complex in the general nonseparable setting \eqref{eq:struct_IV}, which I focus on now. I briefly start with the basic idea of \emph{control variables}, which captures the same information as the unobservable $U$. This allows one to identify average effects \citep{imbens2009identification}. I show how optimal transportation can be used to explore the limitations of this approach and provide a novel generalization. Then I show how arguments using fixed-point iterations \citep{torgovitsky2015identification, d2015identification} allow to identify heterogeneous causal effects. I then see how far these can be extended to general dimensions and what the limitations are \citep{gunsilius2023condition}. In the process I introduce some interesting dynamic mathematical properties of optimal transport maps that warrant further investigation. Then I turn to the partial identification case \citep{balke1994counterfactual,balke1997bounds}, where the goal is to obtain bounds on an average treatment effect. I argue that there are also some so far unexplored connections to optimal transportation in this setting that can be valuable to analyze. Finally, I want to mention the connection of distributionally robust optimization approaches \citep{blanchet2019robust} to instrumental variables estimation, which have recently been explored in linear models \cite{qu2024distributionally}.

\subsection{Control variables}
From the classical backdoor criterion \citep{pearl2009causality}, we know that conditioning on $W$ breaks the backdoor and allows us to extract the correct causal effect $X\to Y$: fixing \( W \) by conditioning on a realization $W=w$ makes \( X = h(Z,w) \) just a function of the exogenous \( Z \), so that by varying \( Z \) one can obtain the exogenous effect of \( X \) on \( Y \) \revision{since $Z$ is assumed to not directly influence $Y$}. The issue is that $W$ is unobservable, so that conditioning on it is impossible. 
\revision{A motivating example for this setting is gasoline demand \citep{imbens2009identification}: the endogenous variable is total household expenditure, the instrument is gross earnings, and the outcome is the gasoline budget share. The nonseparable model is needed because different households, indexed by different realizations of $W$, respond differently to changes in income, heterogeneity that a standard linear IV model cannot capture.} This is where the control variable approach comes in \citep{imbens2009identification}. 
\begin{definition}\label{def: control_var}
A random variable $R:\Omega\to \mathbb{R}^d$ is a control variable in the model \eqref{eq:struct_IV} if $X\perp V | R$.
\end{definition}

A control variable mimics the conditioning on $W$. Therefore, it needs to be constructed such that conditioning on $W$ is the same as conditioning on $R$, i.e., their induced $\sigma$-algebras need to coincide. Doob's functional representation \revision{\citep[e.g.][Lemma 1.1.3]{kallenberg1997foundations}} implies that their induced $\sigma$-algebras coincide if there exist maps $\xi,\rho$ such that $R=\xi(W)$ and $W=\rho(R)$. In particular, these maps \emph{may not} depend on any other variables $Y,Z,V$. \revision{The authors who introduced this idea in \cite{imbens2009identification}}, require \( W \) and \( X \) to be univariate and \( h(z,W) \) to be the monotone rearrangement  in \( W \) for all $z$. Then they construct a control variable via 
\begin{equation}\label{eq:control_fct}
    R = F_{X|Z}(X) = F_W(h^{-1}(Z,X)) = F_W(W).
\end{equation}
They then assume \( F_W \) to be itself strictly increasing and continuous, that is, invertible, which makes \( R \) a valid control variable with the fixed function $F_W$. 

Using this control variable approach, one can now identify average structural effects $\Lambda(X\to Y)$ of the form 
\[\Lambda(X\to Y) = \int \Lambda(g(x,v)) dP_V(v),\] which are unobservable because $P_V$ is unobservable. \revision{Note that unlike in Section 3, the control variable approach does not identify the full mechanism $g(x,v)$, only averages over the unobservable $V$. This is because the endogeneity of $X$ prevents direct access to the counterfactual marginals; identifying the full mechanism requires additional structure, as discussed below.}

\revision{To identify the average structural effects, one must be able to 
integrate over the full distribution of the control variable $R$. This 
requires the \emph{large support assumption}:}
\revision{\begin{definition}[Large support assumption]
For all realizations $x$ of $X$, the support of the conditional 
distribution $P_{R|X=x}$ equals the support of the marginal 
distribution $P_R$.
\end{definition}}

\revision{This condition is related to, but distinct from, the 
overlap (positivity) assumption familiar in causal inference. Overlap 
requires that each unit has positive probability of receiving each 
treatment level; the large support assumption instead requires that the 
instrument generates sufficient variation in the control variable $R$ 
at every value of the endogenous variable $X$. It is a 
completeness-type condition on the instrument's strength, and is 
typically violated when the instrument has finite support 
\citep{torgovitsky2015identification}.}

\revision{With this definition} one can compute
\begin{align*}
     &E(\Lambda(Y) | X = x, R=r) \\
    = &\int \Lambda(g(x, v)) dF_{V | X=x,R=r}(v)\\
    = &\int \Lambda(g(x, v)) dF_{V | R=r}(v) .
\end{align*}
where the second equality follows from the fact that $R$ is a control variable. Then by the large support assumption on $R$ one can integrate over the marginal distribution of $R$ to get
\begin{align*}
\int &E(\Lambda(Y) | X = x, R=r)dF_R(r) \\
= &\iint \Lambda(g(x, v)) dF_{V | X=x,R=r}(v)dF_R(r)\\
    = &\iint \Lambda(g(x, v)) dF_{V | R=r}(v)dF_R(r) \\
    =& \int \Lambda(g(x, v)) dF_{V}(v).
\end{align*}
As can be seen from the above argument, if the large support assumption does not hold, one can only identify the effect on the set where the supports of $P_{R|X=x}$ and $P_R$ overlap for all $x$, which can be significantly smaller. 

The large support assumption is testable in practice and is usually very badly violated \citep{heckman1990varieties}, not least since many instruments have finite support \cite{torgovitsky2015identification}. In particular, the fact that the control variable is by construction uniform is a restricting factor.

The knowledge that this identification result is based on optimal transportation allows us to characterize the full class of valid control variables by replacing the map \( F_{X|Z}(X) \) by a more general function \( m^{-1}(X,Z) \) to generate an $R$ which need not have a uniform distribution. The important requirement on \( m^{-1} \) is that it can be written as \( m^{-1}(X,Z) = T(h^{-1}(Z,X)) \) for some measure-preserving isomorphism \( T \), which \textit{must not depend on} \( z \). A measure-preserving isomorphism $T$ is a map that is measurable and preserves the measure whose inverse exists and is also measurable and measure-preserving. In this case, the two $\sigma$-algebras induced by $R$ and $W$ coincide, so that $R$ is a valid control variable. 
Using this idea, we have the following simple generalization of the control variable approach.  

\begin{proposition}[Generalized control variables] \label{prop:gen_ctrl_var} Let \( F_W \) be absolutely continuous and strictly increasing and let \(h(z,w)\) be the monotone rearrangement between $F_W$ and $F_{X|Z=z}$ for all $z$. If \( F_{X|Z=z} \) is continuous in $x$ for all $z$, then any univariate random variable \( R \) independent of \( Z \) for which there exists a measure-preserving isomorphism $T:[0,1] \to\mathbb{R}$ mapping the uniformly distributed $\tilde{R} = F_{X|Z}(X)$ to $R$ can be made a control variable by setting
\[m^{-1}(X,Z) = T(F_{X|Z}(X)).\]
\end{proposition}
\begin{proof} The first part is the same as the proof of Theorem 1 in \cite{imbens2009identification}. Let \( h^{-1}(x, z) \) denote the inverse function of \( h(z, w) \). Then, by \eqref{eq:matzkin_argument}, it holds
\[
F_{X|Z=z}(x) =  F_{W}(h^{-1}(x, z)).\] Plugging in the random variables $X,Z$ into this expression gives
\[
\tilde{R} = F_{X|Z}(X) = F_{W}(h^{-1}(Z, X)) = F_{W}(W),
\]
where $\tilde{R}$ is now a uniformly distributed random variable, i.e., $\tilde{R}\sim U[0,1]$. Now fix a random variable $R$ for which there exists a measure-preserving isomorphism $\tilde{R}\mapsto T(\tilde{R})=R$. Set 
\[m^{-1}(X,Z) = T(F_{X|Z}(X)),\] which implies 
\begin{align*}
    R = m^{-1}(X,Z)=&
     T\left(F_W(h^{-1}(Z,X))\right) \\
    =& T\left(F_W(W)\right).
\end{align*}

Since $F_W$ is strictly increasing and continuous, it is the unique monotone rearrangement between $W$ and $\bar{R}$, and is measurable with measurable inverse since both $\bar{R}$ and $W$ do not give mass to points. In particular, $F_W$ does not depend on $X$ or $Z$. But the composition $T\circ F_W(W)$ is then a measure-preserving isomorphism that does not depend on $X$ or $Z$.
Hence, the $\sigma$-algebra induced by \( W \) is equal to the $\sigma$-algebra induced by \( R \), so that conditional expectations given \( W \) are identical to those given \( R \). Also, for any bounded function \( a(X) \), by $Z\perp (V,W)$
\[
E[a(X) | V,W] = \int a(h(z, W)) dF_Z(z) = E[a(X) | W].
\]
Therefore, for any bounded function \( b(V) \), I have
\begin{align*}
E[a(X) b(V) | R] =& E[a(X) b(V) | W]\\
=&E[b(V) E[a(X) | V,W] | W]\\
=& E[b(V) E[a(X) | W] | W]\\
=& E[b(V) | W] E[a(X) | W]\\
=& E[b(V) | R] E[a(X) | R]
\end{align*}
which shows that $R$ is a control variable.
\end{proof}
Proposition \ref{prop:gen_ctrl_var} is a straightforward extension of the control variable approach once one realizes that the identification idea is again based on the monotone rearrangement. \revision{Since $T$ preserves the $\sigma$-algebra, the choice of $T$ affects neither identification nor first-order asymptotic efficiency; however, when $F_{X|Z}$ is estimated, $T$ enters the higher-order bias of the generated regressor, so that the choice may matter in finite samples.}

A generalization of Proposition \ref{prop:gen_ctrl_var} to multivariate settings is essentially impossible since it is no longer the case that the half-open rectangles $(-\infty,w]$ and $(-\infty,w']$ necessarily have different probabilities if $w\neq w'$ for strictly increasing $F_W$, due to the lack of a natural complete order on $\mathbb{R}^d$. For instance, under the assumption that $h(z,w)$ is the gradient of some convex function \revision{($h(z,w) := \nabla\varphi_z(w)$),} the important requirement 
\[
P_W(h^{-1}(z, A)) = P_{X|Z=z}(A) = P_R(m^{-1}(z, A)) 
\]
does not hold for arbitrary Borel sets $A\subset \mathbb{R}^d$ and all $z$ unless $R=W$, in which case the control variable approach is obsolete. The proof for this claim is  omitted. Hence, one of the most influential approaches to identify causal effects in general instrumental variables seems to be restricted to univariate first stages; it would be interesting to investigate this further. Note, however, that there is no restriction on the second stage $Y=g(X,V)$, the stage of interest. 

\subsection{Instruments with small support and dynamics of Brenier maps}
As mentioned, the large support assumption is by definition violated when the instrument $Z$ is not continuous \citep{torgovitsky2015identification}. In the following, I consider the case where Z is binary and can only take two values $z, z'$; the result can be straightforwardly extended to finitely many realizations. Such settings are ubiquitous in practice \citep{imbens2007nonadditive, torgovitsky2015identification}. Moreover, it would be nice to allow for several endogenous variables $X$ and not just one, meaning that it is useful to generalize the first stage in \eqref{eq:struct_IV} to a multivariate setting. Finally, the control variable approach only provided identification for average structural effects $\Lambda(X\to Y)$, but not the entire mechanism $g(X,V)$, which would give us full identification of the heterogeneous treatment effects. 

To incorporate all of these modifications, we again start with the case where all variables are univariate. This has been introduced in \cite{torgovitsky2015identification} and \cite{d2015identification}. The univariate model is the one where all variables, observable and unobservable, are univariate and where both \( g \) and \( h \) are assumed to be monotone rearrangements in $V$ and $W$, respectively. Using the counterfactual notation \citep{rubin1974estimating}, we can write $F_{Y(X)}$ as the counterfactual distribution of the effect of $X$ on $Y$ for an \emph{exogenous} change in $X$ (i.e., the counterfactual distribution). Due to the backdoor channel via $V$ as depicted in Figure \ref{fig:DAG_IV}, this is not observable and the observable distribution $F_{Y|X}$ does not coincide with the counterfactual distribution. 
 
 The idea for identification of $g$ is as before: vary \( Z \) in such a way that \( X \) varies but \( V \) stays constant.
 The fact that this is possible if \( Z \) is only binary was first shown in \cite{torgovitsky2015identification} and \cite{d2015identification}. The argument of the former is more amenable to our ideas and rests on the fact that using the binary instrument \( Z \) with realizations \( z \) and \( z' \), there are two maps that do not change the distribution \( F_{V} \) of \( V \), but change values of \( X \). Using only those two maps therefore captures the exogenous effect of \( X \) on \( Y \). These two maps are depicted in Figure~\ref{fig:torgovitsky}, which is taken from \cite{gunsilius2023condition}.
 \begin{figure}[h!]
\centering
\begin{tikzpicture}[scale=0.7]
\draw[->, thick] (0,0) to (0,4);
\draw[->,thick] (0,0) to (4,0);
\draw[-,thick] (0,0) to [out=15,in=180] (4,4);
\draw[thick] (0,0) to [out=90,in=220] (4,3.5);
\draw node[right] at (4,0) {$x$};
\draw[-,thick] (0.5,0) to (0.5,-0.1);
\draw node[below] at (0.5,-0.1) {$x_0$};
\draw[-,dashed] (0.5,0.2) to (0.5,1.38);
\draw[-,thick] (0,0.2) to (-0.1,0.2);
\draw node[left] at (-0.1,0.2) {$F_{X|Z=z}(x_0)$};
\draw[-,dashed] (0.52,1.38) to (1.5,1.38);
\draw[-,thick] (1.55,0) to (1.55,-0.1);
\draw node[left] at (-0.1,1.32) {$F_{X|Z=z}(Tx_0)=F_{X|Z=z'}(x_0)$};
\draw node[below] at (1.55,-0.1) {$Tx_0$};
\draw[-,thick] (0,1.32) to (-0.1,1.32);
\draw[-,dashed] (1.55,1.38) to (1.55,2.1);
\draw[-,dashed] (1.55,2.1) to (1.9,2.1);
\draw[-,dashed] (1.9,2.1) to (1.9,2.28);
\draw[-,thick] (0,2.1) to (-0.1,2.1);
\draw node[left] at (-0.1,2.1) {$F_{X|Z=z}(T^2x_0)=F_{X|Z=z'}(Tx_0)$};
\draw node[right] at (1.95,-0.25) {$x^*$};
\draw[-,dotted, thick] (2.05,0) to (2.05,2.4);
\draw[-,thick] (2.05,0) to (2.05,-0.1);
\draw node[right] at (2.3,0.8) {$F_{X|Z=z}$};
\draw[->,thick] (2.3,0.9) to (1.43,1.17);
\draw node[below] at (3.6,2.4) {$F_{X|Z=z'}$};
\draw[->,thick] (3.4,2.4) to (3.45,3);
\end{tikzpicture}
\caption{Fixed-point iteration in a univariate framework for identifying causal effects in \eqref{eq:struct_IV}.}
\label{fig:torgovitsky}
\end{figure}
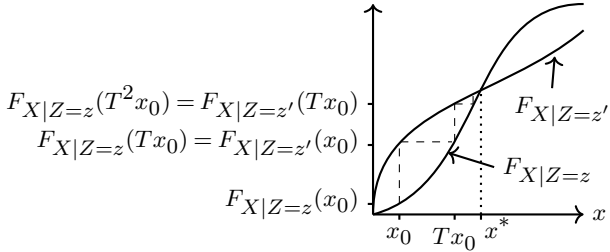

The first map changes \( z \mapsto z' \) for a fixed \( x \), i.e. switches the distributions \( F_{X|Z=z}(x) \) and \( F_{X|Z=z'}(x) \) (the ``vertical'' map in Figure~\ref{fig:torgovitsky}). The fact that \( F_{V|X,Z} \) is not affected by this follows from the control variable approach introduced above \citep{imbens2009identification} and because \( Z \) has no effect on \( g \) due to the exclusion restriction of \( Z \). A control variable \( R \) can be constructed via \( R = F_{X|Z}(X) \). Hence, by conditioning on \( R \) and the fact that \( Z \) is independent of \( (V, W) \), the vertical shift changes \( Z \) but does not change \( X \), which therefore does not affect the function \( g(x,v) \) we want to identify \citep{torgovitsky2015identification}.

The second map is the change of quantiles (the ``horizontal'' map in Figure \ref{fig:torgovitsky}), which follows by the fact that the change \( (x, z') \to (Tx, z) \) is performed in such a way that \( F_{X|Z=z'}(x) = F_{X|Z=z}(Tx) \). This map is of course the monotone rearrangement if it exists. This is again achieved via the control variable approach by defining \( R = F_{X|Z}(X) \) and conditioning on $R$. This implies that the horizontal map does not affect the distribution of \(V \) and hence the function \( g(\cdot, V) \) we want to identify. If we keep alternating between these two maps for a given starting value \( x_0 \), this sequence 
\(\lim_{m\to+\infty} T^m(x_0)\), that is, \( x_0, T(x_0), T(T(x_0)), T(T(T(x_0))) \dots\) either converges or diverges. 

To make it converge, one now assumes the existence of a fixed point $x^*$, for instance by requiring that the conditional CDFs $F_{X|Z=z}$ and $F_{X|Z=z'}$ are continuous and that they intersect at a point, as depicted in Figure \ref{fig:torgovitsky} \citep{torgovitsky2015identification}. In this case, the dynamics will converge to \( x^* \) where \( F_{X|Z=z} \) and \( F_{X|Z=z'} \) intersect. For points \( x \leq x^* \) we need to iterate \( z \mapsto z' \) and for points \( x \geq x^* \) we need to iterate \( z' \mapsto z \). This allows us to identify the function \( g(x, V) \) by comparing different points \( x \) in this iterative approach to the point \( x^* \), because the ``vertical'' and ``horizontal'' maps do not change the distribution of \( V \) and hence keep \( g(\cdot, V) \) fixed. Now in combination with the standard assumption that $g$ is the monotone rearrangement in $V$ and some other regularity assumptions, this argument allows one to identify the true causal mechanism $g$, but now in the setting $X\not\perp U$. 

\revision{To see why this yields identification more concretely, note that at the fixed point $x^*$ the two conditional CDFs coincide, $F_{X|Z=z}(x^*)=F_{X|Z=z'}(x^*)$, so the control variable $R=F_{X|Z}(X)$ takes the same value regardless of $Z$. This means that $X$ is effectively exogenous at $x^*$, and the mechanism $g(x^*,V)$ is identified directly from the observed conditional distribution $P_{Y|X=x^*}$ (under a normalization). Since each vertical and horizontal step preserves the distribution of $V$, it follows that $g(x,V)$ is identified at every $x$ reached by the iteration. The convergence of the iteration to $x^*$ then propagates identification from $x^*$ to all other values of $x$.}

The mathematically interesting part of this argument is the fixed-point iteration. In particular, it is an interesting question if this argument can be extended to a multivariate first stage and what the dynamics look like in this case. This has been done in \cite{gunsilius2023condition} in a special case of optimal transport maps, but there are many other open questions. The question then becomes: is it possible to generalize the fixed-point idea from the univariate setting when considering dynamics of Brenier maps $\nabla\varphi_z$? The first key is to generalize the ``horizontal'' map in Figure \ref{fig:torgovitsky} to the multivariate setting (the ``vertical'' \revision{map is} the same as in the univariate setting). 

The main requirement pointed out in \cite{gunsilius2023condition} is that $h(z,\cdot)$ be a measure-preserving isomorphism for every $z$, i.e., a map that preserves the measure whose inverse also preserves the measure in the sense that $P_{X|Z=z}(A) = P_W(h^{-1}(z,A))$ for any Borel set $A$, and where $h^{-1}(z,A)$ denotes the preimage of $A$. This is because for any measure-preserving isomorphism, it holds
\begin{equation}\label{equalitiesimportant}
P_{V|X=x,Z=z}=P_{V|W=h^{-1}(z,x),Z=z}=P_{V|W=h^{-1}(z,x)},
\end{equation}
where the second equality follows since $Z$ is an instrument. The first equality holds by defining the map $\phi: (X,Z) \mapsto (h^{-1}(Z,X),Z)$, which is a measure-preserving isomorphism. 
The same reasoning holds for the map 
\[(v,x,z)\mapsto (v,h^{-1}(z,x),z).\] 
Thus $P_{V|X,Z}(A_v)=P_{V|W,Z}(A_v).$
This shows that $V$ does not change when applying the map $\phi$, which makes the latter the natural extension of the ``horizontal map'' in the univariate setting.
This argument shows that a measure preserving isomorphism is the required generalization of the ``horizontal map''. 

As a result, this shows that if we make some standard regularity assumptions, such as continuity of $h(z,\cdot)$, \emph{then $W$ must be of the same dimension as $X$.} This necessary condition has been well-understood and has also been found in \cite{hoderlein2017corrigendum}. In particular, this implies that identification in general instrumental variable models of the form \eqref{eq:struct_IV} puts a strong limit on the dimension on $W$ in comparison to $X$---in short, the dimension of the unobservable in the first stage needs to be restricted, while the second stage can be much more general.\\

 \noindent\emph{Multivariate extension.}
  Generalizing the ``horizontal'' map from the univariate case to the multivariate setting is challenging. A multivariate construction based on the Brenier map transporting $P_{X|Z=z}$ back to $P_{X|Z=z'}$ is developed in
  \cite{gunsilius2023condition}. The relevant fixed object need not be a single point: it is the fixed \emph{set} of this map. Its geometry and the behavior of the associated orbits are determined by the resulting Brenier-map dynamics.

  The required regularity and nondegeneracy conditions are summarized in the following assumption.

  \begin{assumption}\label{ass:rankmap}
   Let $R_z:\mathcal{X}\to\mathcal{X}$ denote the Brenier map transporting
  \[
  P_z:=P_{X|Z=z}
  \quad\text{to}\quad
  P_{z'}:=P_{X|Z=z'},
  \]
  and define its fixed set by
  \[
  \mathcal{I}_R:=\{x\in\mathcal{X}:R_z(x)=x\}.
  \]
  The conditional laws $P_z$ and $P_{z'}$ have a common compact, uniformly convex support $\mathcal{X}\subset\mathbb{R}^k$ with $C^2$ boundary. They are absolutely continuous, with densities
  \[
  f_z,f_{z'}\in C^{0,\alpha}(\overline{\mathcal{X}})
  \]
  for some $\alpha\in(0,1)$, and there exist constants
  $0<\lambda<\Lambda<\infty$ such that
  \[
  \lambda\leq f_z(x),f_{z'}(x)\leq\Lambda
  \quad
  \text{for every }x\in\overline{\mathcal{X}}.
  \]
  After fixing a $C^1$ extension $\widetilde R_z$ to an open neighborhood of
  $\overline{\mathcal{X}}$, there exist an open neighborhood
  $U\supset\mathcal{I}_R$ and an integer $r\in\{1,\ldots,k\}$ such that
  \[
  \operatorname{rank}\!\left(D\widetilde R_z(x)-I_k\right)=r
  \quad
  \text{for every }x\in U.
  \]
  \end{assumption}

The support and density requirements in Assumption \ref{ass:rankmap} are standard sufficient conditions from Caffarelli's global regularity theory for the Monge--Amp\`ere equation \citep{caffarelli1990interior,caffarelli1992regularity,caffarelli1996boundary}.They imply that the Brenier potential $\psi$ admits a $C^{2,\alpha}$ representative up to the boundary and hence that the rank map $R_z=\nabla\psi$ admits a $C^{1,\alpha}$ representative  \citep[see, e.g.,][]{villani2021topics}. Compactness and convexity also allow one to apply Brouwer's fixed-point theorem and provide the compactness needed for the orbit-accumulation argument. These conditions do not impose unimodality and therefore admit smooth multimodal and non-quasi-concave distributions. The constant-rank requirement does not follow from Caffarelli's theory; it is a separate nondegeneracy condition ensuring that the rank fixed set is locally lower-dimensional and hence has Lebesgue measure zero. These conditions are sufficient rather than necessary, and extending the argument beyond common compact supports and the stated regularity conditions seems fruitful.

  To connect this transport map to the structural first stage, maintain the assumption $Z\perp(V,W)$, assume that $P_W$ is absolutely continuous, and assume that $h(\tilde z,\cdot)$ is cyclically monotone for each $\tilde z\in\{z,z'\}$. Impose the first-stage normalization
  \[
  h(z',w)=w,
  \]
  and write $h_z(w):=h(z,w)$. Since $X=h(Z,W)$, this normalization gives
  \[
  P_{z'}=P_W
  \qquad\text{and}\qquad
  (h_z)_{\#}P_{z'}=P_z.
  \]
  Brenier uniqueness therefore implies that $h_z$ coincides
  $P_{z'}$-almost everywhere with the unique Brenier map
  $T:\mathcal{X}\to\mathcal{X}$ satisfying
  \[
  T_{\#}P_{z'}=P_z.
  \]
  Its $P_z$-almost-everywhere inverse coincides with the map $R_z$ introduced
  above:
  \[
  R_z=T^{-1}=h_z^{-1}
  \qquad P_z\text{-almost everywhere}.
  \]
  Thus $R_z$ sends an observed value of $X$ in regime $z$ back to its
  first-stage rank $W$ in the normalized regime $z'$. It therefore provides
  the multivariate analogue of the ``horizontal'' map.

  There is also a useful connection between this recursion and gradient dynamics. Write $R_z=\nabla\psi$ for its Brenier potential and define
  \[
  L(x):=\psi(x)-\frac{1}{2}\|x\|^2.
  \]
  The rank-map recursion can then be written as
  \[
  x_{n+1}=x_n+\nabla L(x_n).
  \]
  It is therefore unit-step gradient ascent, or equivalently the explicit-Euler discretization with step size one of the continuous gradient-ascent flow
  \[
  \dot{x}=\nabla L(x).
  \]
  Moreover,
  \[
  \mathcal{I}_R
  =
  \{x\in\mathcal{X}:R_z(x)=x\}
  =
  \{x\in\mathcal{X}:\nabla L(x)=0\},
  \]
  and convexity of $\psi$ gives the Lyapunov inequality
  \[
  L(x_{n+1})-L(x_n)
  \geq
  \frac{1}{2}\|x_{n+1}-x_n\|^2.
  \]

  \begin{claim}\label{claim:rank_map_dynamics}
  Under Assumption \ref{ass:rankmap}, the inverse Brenier rank map $R_z$ admits a $C^{1,\alpha}$ representative and its fixed set $\mathcal{I}_R$ is nonempty. Moreover, $\mathcal{I}_R$ is locally contained in the intersection of $\mathcal{X}$
  with a $C^1$ embedded submanifold of $\mathbb{R}^k$ of dimension $k-r<k$, and hence has Lebesgue measure zero. For every $x_0\in\mathcal{X}$, every accumulation point of the orbit
  \[
  \{R_z^n(x_0)\}_{n\geq0}
  \]
  belongs to $\mathcal{I}_R$, and
  \[
  \operatorname{dist}\!\left(R_z^n(x_0),\mathcal{I}_R\right)
  \longrightarrow 0.
  \]
 This establishes convergence to the fixed set in distance. The preceding argument alone does not establish convergence of the full orbit to a single point.
  \end{claim}

  The rank-map transition preserves the first-stage rank $W$. Since $Z\perp(V,W)$, for $P_z$-almost every $x$,
  \[
  P_{V|X=x,Z=z}
  =
  P_{V|X=R_z(x),Z=z'}.
  \]
  This equality allows the identification argument to be propagated along rank-map orbits. In particular, let $C\subset\mathcal{I}_R$ be a closed connected piece of the fixed set and regard $C$ as one quotient anchor. Call $C$ normalized only when the second-stage normalization fixes the observational-equivalence ambiguity at every point of $C$ that can arise as an orbit accumulation point, unless the maintained assumptions make that ambiguity constant on $C$. Define its accumulation basin by
  define its accumulation basin by
  \begin{multline*}
  \mathcal{B}(C)
  :=
  \left\{
  x\in\mathcal{X}\setminus\mathcal{I}_R:
  \text{some accumulation}\right.\\
  \left. \text{point of }
  \{R_z^n(x)\}_{n\geq0}
  \text{ belongs to }C
  \right\}.
  \end{multline*}
Under the maintained conditions that $g(x,\cdot)$ is uniquely identified in the exogenous-$X$ problem, that $g(\cdot,v)$ is uniformly continuous in probability in $x$, that $P_{V|X=x}$ is absolutely continuous, and that the support of $V$ is convex and independent of $X$ and $Z$, the mechanism $g(x,v)$ is identified for $P_X$-almost every $x\in\mathcal{B}(C)$ and $P_V$-almost every $v$. If the accumulation basins of countably many normalized fixed-set pieces cover $\mathcal{X}\setminus\mathcal{I}_R$ up to a $P_X$-null set, then $g(x,v)$ is identified for $P_X$-almost every $x$ and $P_V$-almost every $v$ \citep{gunsilius2023condition}.

The interesting part of this identification result lies in the connection between optimal transport, fixed-point methods, and discrete gradient dynamics. The inverse Brenier rank map generates a dynamical system whose fixed set provides the lower-dimensional anchor for identification. Although the assumptions above describe the accumulation behavior of its orbits, finer properties of these dynamics---including full-orbit convergence, stability, and the geometry of the accumulation basins---remain open in more general settings. These questions are not specific to causal inference: fixed-point and gradient iterations arise throughout mathematics, for instance in the analysis of convergence to equilibria \citep[e.g.][]{mascolell1995micro}. Extending the rank-map argument to distributions with different or unbounded supports, less regular densities, or more general transport maps seems fruitful in this regard.

\subsection{Bounds on average outcomes in instrumental variable models}\label{sec:IV_bounds}
The above results and connections between optimal transport and \revision{causal inference were straightforward} in that they were all based on the monotone rearrangement. While the monotone rearrangement will make another appearance later, for now we explore a less straightforward connection between optimal transport and causal inference. It is still situated in the instrumental variable framework, but instead of trying to identify the full causal mechanism $g(X,V)$, we now only care about average structural effects, similar to the control variable approach above. As we have seen, the control variable approach requires strong structural assumptions on the causal mechanism $g$ and the first stage $h$ in order to identify those average structural effects.

In this section, we therefore go the opposite way: instead of making strong structural assumptions on the mechanisms, we ask how tight the bounds on the object of interest will be if we make \emph{no assumptions} on the structure. This mirrors \cite{balakrishnan2023conservative}, but in the setting of instrumental variables, where there is more structure. Also, as will become clear, this setup is not directly related to the classical optimal transport problem we have considered so far, but more general---it is a version of an optimal transport problem on path spaces. I want to cover it in the hope that this connection will bring mathematically novel insights into not just causal inference, but also optimal transport. 

As in the control variables approach, our goal is to identify average structural effects. Before defining them, I want to recall the connection between structural equations, the counterfactual notation \citep{rubin1974estimating}, and ``counterfactual processes'' \citep{balke1994counterfactual,pearl2009causality}: the structural equation $Y=g(X,V)$ is equivalent to the counterfactual process $Y_x(v)$, and analogously for $X=h(Z,W)$ and $X_z(w)$. This allows us to obtain the bounds by optimizing over the \emph{counterfactual path space}, which is a convenient representation in many settings. 

The original idea for this approach was introduced in \cite{balke1994counterfactual,balke1997bounds} in the case where all observable variables $Y,X,Z$ are binary (see also the contributions in \cite{manski1990nonparametric}). These results were generalized to continuous $Y$ and binary $X$, and $Z$ in \cite{kitagawa2021identification} and to discrete $Y$, $X$, and $Z$ in \cite{russell2021sharp}. These partial identification results are part of a vast literature in econometrics and causal inference, see \cite[e.g.][]{manski1999identification,manski2003partial}. In the following we work in the general framework of $Y$,$X$, and $Z$ having potentially continuous laws \citep{gunsilius2019path}. 

Our goal here is to obtain bounds on average structural effects of the form
\[\Lambda(X\to Y) = \mathbb{E}[\Lambda(g(x,V))]\] as in the case of control variables above. The seminal idea of \cite{balke1994counterfactual,balke1997bounds} is to consider the counterfactual distributions $P_{Y(x)}$ and $P_{X(z)}$ as the laws of corresponding counterfactual processes $Y_x(v)$ and $X_z(w)$ of the first and second stage of the IV model \eqref{eq:struct_IV}. Each element $v\in\mathcal{V}$ indexes one path $Y_x(v)$ and each $w\in\mathcal{W}$ indexes one path $X_z(w)$ of the processes, respectively. Mathematically, this is possible if the spaces of paths are small enough, so that they can be put in a one-to-one relation to the unit interval, for instance the space of all continuous paths \citep[e.g.][Theorem 9.2.2]{bogachev2007measure}.
The average structural effect can then be written as
\[\Lambda(X\to Y) = \mathbb{E}[\Lambda(Y_x)],\] where the expectation is taken with respect to the joint distribution $(P_V, P_W)$ over the \revision{latent variable $U\equiv (V,W)$}. In this setting, it is often easier to consider the unobservable $U$ instead of $(V,W)$, but both representations are equivalent. 

The idea for obtaining bounds on $\Lambda(X\to Y)$ is the following. As mentioned in the introduction, because of the endogeneity problem $X\not\perp U$, the observable distribution $P_{Y|X=x}$ is \revision{\emph{not}} the correct counterfactual law $P_{Y(x)}$. However, \emph{if $Z$ is a valid instrument}, then the observable conditional distribution $P_{Y,X|Z=z}$ does provide correct information on the causal system. So the key is to back out bounds on $P_{Y(x)}$, which is required for $\Lambda(X\to Y)$, from the observed distribution $P_{Y,X|Z=z}$. One can achieve this via a linear program, which is also the connection to optimal transport---and the reason I include it in this review. 

The linear program to obtain bounds on $\Lambda(X\to Y)$ is
\begin{equation}\label{eq:bound_equation}
\begin{aligned}
&\underset{\substack{P_U\in\mathcal{P}^*}}{\min/\max}\quad  \mathbb{E}_{P_U}[\Lambda(Y_{x})]\\
\text{s.t.}\quad  &F_{Y,X|Z=z}(y,x) = P_U(Y_{X_z}\leq y,X_z\leq x),
\end{aligned}
\end{equation}
where $\mathcal{P}^*$ is a set of probability measures on the joint path space of the processes $Y_x(u)$ and $X_z(u)$, and $Y_{X_z(u)}(u)$ is the process based on a realization of the process $X_z(u)$ for given $u$. In words, the optimization problem tries to find a corresponding measure $P_U$ on path space which maximizes (for an upper bound) or minimizes (for a lower bound) the average structural effect $\Lambda(X\to Y)$ under the restriction that $P_U$ induces counterfactual processes $Y_x$ and $X_z$ whose induced joint law
\[F_{[Y,X]_z}= P_U(Y_{X_z}\leq y,X_z\leq x)\] coincides with the \emph{observable} marginal  $F_{Y,X|Z=z}$. 

The problem \eqref{eq:bound_equation} is a generalized optimal transport problem on path spaces, similar to the \revision{\emph{stochastic optimal transport problem} \citep{mikami2021stochastic}}. To see the connection, suppose that $Z$ can only take two values, $z$ and $z'$. In this case, \eqref{eq:bound_equation} requires to find a joint measure for $([Y,X]_z,[Y,X]_{z'})$ such that the marginals coincide with the observable marginal measures $P_{Y,X|Z=z}$ and $P_{Y,X|Z=z'}$. For more general $Z$, it becomes the problem on path measures. It is very closely related to---but more general than---problems from statistical physics, in particular large particle dynamics \citep[e.g.][]{dawsont1987large, follmer1988random, cattiaux1995large}, \revision{the HWI inequality \citep{gentil2020entropic}, and Csiz\'ar projections \citep{csiszar1975divergence}.} The considered problem is more general, because the objective expression only depends on $Y$ and $X$, while the marginals are for $Y,X,$ and $Z$. \revision{These connections seem to be extremely tantalizing and fruitful for future research at the intersection of optimal control, causal inference, and statistical physics, with optimal transport being the bridge.}

So far, only a rather inefficient sampling method has been proposed to solve this problem in practice \citep{gunsilius2019path}, for a very general set of processes $Y_x, X_z$. \cite{kilbertus2020class} provide a clever efficient estimator by simplifying the problem: instead of requiring a replication of the entire marginal distribution in the constraint, they essentially replicate moments of the data, turning the problem into a generalized method of moments problem. While this provides a robust and efficient estimator, it does not solve the original causal inference problem. \revision{A promising direction is entropic regularization: by adding a relative entropy penalty, the infinite-dimensional linear program can be transformed into a multi-marginal Schr\"odinger bridge problem, which admits tractable approximations with convergence guarantees \citep{gunsilius2025schrodinger}. This approach is the first that makes the solution to \eqref{eq:bound_equation} feasible, statistically and computationally.}

Furthermore, this setup seems like a useful generalization of classical optimal transport problems to measures on path spaces, whose analysis is likely to provide new insights and connections in the mathematical theory of optimal transportation. In the finite setting, i.e., where $X$ and $Z$ are supported on finitely many points, the problem reduces to a simple finite-dimensional linear program, which has been analyzed recently in statistics \citep{klatt2022limit} and econometrics \citep{fang2023inference}.

\subsection{Distributionally robust IV methods}
Before moving on to other general approaches for identification of causal effects in settings with observable data, I want to mention another area where optimal transportation arguments are paramount, and which has recently gained attention in the literature on causal inference: distributionally robust optimization (DRO) \citep{blanchet2019robust, blanchet2019quantifying, gao2023distributionally}. 

Here, the object of interest is $\mathbb{E}[l(D,\beta)]$ for some general loss function $l(D,\beta)$, where the expectation is taken with respect to the law of the data $D$ in an i.i.d.~setting. $\beta$ is the parameter one wants to optimize over. The idea of DRO is to make the estimator robust to different heterogeneous environments, as encoded by the law from the data $D$ is drawn. This heterogeneity is captured by defining a region of possible distributions the problem can take around the data-distribution. The utility of DRO comes from the fact that when one chooses a Wasserstein ball for the regions of possible distributions, the primal DRO problem
\[\min_\beta \max_{Q\in B_\rho(P)} \mathbb{E}_Q[l(D,\beta)],\] where 
\(B_\rho(P) = \{Q : W_2(Q,P)\leq \rho\}\) is a ball in the Wasserstein space of radius $\rho$, admits a dual problem that often takes the form of a standard constrained prediction problem \citep{blanchet2019quantifying, gao2023distributionally}.

Recently, a very interesting contribution \citep{qu2024distributionally} considered the DRO approach for \emph{linear} instrumental variable models, that is where 
\begin{align*}
    g(X,V) &= X^\top\beta_0+ V\quad\text{and}\\
    h(Z,W) &= Z^\top\gamma + W. 
\end{align*}
The optimization problem in this case is 
\[\min_\beta \max_{Q\in B_\rho(\tilde{P}_n)} \mathbb{E}_Q[(Y-X^\top\beta_0)^2],\] where $\tilde{P}_n$ is not the empirical measure for the observations $\{(Y_i,X_i)\}_{i=1}^n$, which are considered i.i.d.~draws from $P$, but the empirical measure of the observations \emph{projected onto the instrument} $Z$, that is,
\[\{(\tilde{Y}_i,\tilde{X}_i)\}_{i=1}^n = \{(\Pi_Z Y_i,\Pi_Z X_i)\}_{i=1}^n,\] where \(\Pi_Z = (Z^\top Z)^{-1} Z^\top\) is the projection onto the columns space of $Z$. 
The ingenuity of this approach lies in the dual problem, which takes a regularized regression form:
\[\min_{\beta} \sqrt{\frac{1}{n} \|\Pi_{\mathbf{Z}} \mathbf{Y} - \Pi_{\mathbf{Z}} \mathbf{X} \beta\|^2} + \sqrt{\rho (\|\beta\|^2 + 1)},
\]
where the interesting part is the novel regularizer, dubbed the ``square root ridge'' regularizer \citep{qu2024distributionally}. In particular, the authors show that the empirical estimator $\hat{\beta}_{n, \rho}$ of $\beta$, is consistent for small enough $\rho>0$, even if $\rho$ does not vanish as the number of data points increases, amid some other nice properties with respect to the weak instrument problem \citep{bound1995problems, stock2002survey} and invalid instruments. Generalizing this idea to the semi-parametric or fully nonparametric setting to understand the corresponding properties and how it compares to the other approaches in instrumental variable models we have considered so far could be useful.

\revision{Throughout this section, all identification arguments have been presented unconditionally. In practice, instruments are often only valid conditional on observed covariates $L$, i.e., $Z \perp (\varepsilon, U) \mid L$. In this case, all of the above arguments go through within strata defined by $L$, with the OT problems solved conditionally. More efficient strategies that exploit the transport structure directly when $L$ is high-dimensional are an interesting direction for future work.}

\section{Using time variation for identification: difference-in-differences and  parallel trends}
\revision{In this section, $t$ indexes time periods, $g\in\{0,1\}$ indexes treatment and control groups, and $Y_{g,t}$ denotes the outcome for group $g$ at time $t$.}

While instruments are a ``silver bullet'' for the general endogeneity problem introduced above, finding a valid and relevant instrument is the challenge in practice. This is amplified by the fact that instrument validity is difficult to test, which turns out to be an impossibility when the instrument $Z$ is distributed on a continuum \citep{gunsilius2021nontestability, pearl1995testability}. 

Moreover, in many settings, other information is available on the problem, often a time domain. The simplest setting is the one with two time periods---a pre-treatment period $t=0$ and a post-treatment period $t=1$---and two groups of interest---usually a treated group $g=1$ and an untreated control group $g=0$. This is the setting in the minimum wage analysis by Card and Krueger \citep{card1994minimum} \revision{mentioned in the introduction} for instance. 

It is natural to exploit the given structure. The idea is to make the \emph{parallel trends assumption}. This \revision{assumption states} that the change over time of the outcome in the control group, that is, the group that did not receive the treatment between $t=0$ and $t=1$, captures all the unobservable influences that cause the outcome to change over time \emph{except for the influence of the actual treatment}. Under this assumption, one can then extract the effect of receiving treatment between $t=0$ and $t=1$ by comparing the change in the outcome over time for the treatment group with the change in the outcome over time for the control group: a difference-in-differences.  

This method of difference-in-difference has become one of the main pillars for reduced-form causal inference in applied research, extended to many time periods and several different techniques. \revision{The standard DiD models are parametric linear,} that is, $g(X,V) = X^\top\beta +V$ and $h(Z,W) = Z^\top\gamma + W$. This allows for simple linear regression methods to identify \emph{average treatment effects}, as in \cite{card1994minimum}. For a recent overview of this applied literature, consider \cite{roth2023s}. A formal treatment of the semiparametric case, still focusing on average effects, can be found in \cite{abadie2005semiparametric}. The methods focusing on average effects are fundamental, but can miss important details as laid out in the minimum wage example in the introduction.

\subsection{Nonlinear difference-in-differences and the changes-in-changes estimator}
A nonparametric method to identify general heterogeneous effects for univariate outcomes was introduced in \cite{athey2006identification}. \revision{It is based on the monotone rearrangement}, and considers the entire distribution of outcomes and unobservables. The abstract relations between all measures are depicted in Figure \ref{fig:DID}, adapted from \cite{torous2024optimal}.
\begin{figure}[h!t]
\centering
\begin{tikzpicture}[scale=0.7]
\draw node[left] at (-0.4,3) {$\textbf{Control}$};
\draw node[right] at (3,3) {$\textbf{Treatment}$};

\draw node[left] at (-3.5,2) {$\textbf{t=1}$};
\draw node[left] at (-3.5,-1) {$\textbf{t=0}$};

\draw node[left] at (-1.7,2.2) {$\boxed{P_{Y_{C,1}}}$};
\draw node at (-0.1,0.5) {$\nu$};
\draw node at (-0.1,-1) {$\nu$};
\draw node[left] at (-1.7,-1.1) {$\boxed{P_{Y_{C,0}}}$};

\draw[->] (-0.25,0.67) to (-1.7,1.9);
\draw[->, dashed] (-0.1,-0.75) to (-0.1,0.25);
\draw[->] (-0.3,-1) to (-1.7,-1);
\draw[->,thick] (-2,-0.75) to (-2,1.75);

\draw node[left] at (-2,0.5) {$\mathrm{d} $};
\draw node[above] at (-0.9,-1) {$h_0$};
\draw node at (-1,0.95) {$h_1$};

\draw node[left] at (2.2,-1.1) {$\boxed{P_{Y_{T,0}}}$};
\draw node[left] at (2.2,2) {$P_{Y_{T,1}^\dagger}$};
\draw node at (4,-1) {$\nu^\star$};
\draw node at (4,0.5) {$\nu^\star$};
\draw node[right] at (5.8,2.1) {$\boxed{P_{Y_{T,1}}}$};

\draw[->] (3.7,-1) to (2.2,-1);
\draw[->, dashed] (3.9,-0.75) to (3.9,0.2);
\draw[->] (3.75,0.6) to (2.2,1.8);
\draw[->] (4.25,0.6) to (5.8,1.8);
\draw[->,thick] (2,-0.75) to (2,1.65);
\draw[->,thick] (5.8,2) to (2.2,2);

\draw node[left] at (2,0.5) {$ \mathrm{d} $};
\draw node[above] at (3,-1) {$h_0$};
\draw node at (3,0.9) {$h_1$};
\draw node at (5.1,0.9) {$h_1^\star$};
\draw node at (4,1.75) {$ \mathrm{T}$};
\end{tikzpicture}
\caption{Illustration of various maps in the ``nonlinear difference-in-differences'' setup. An arrow indicates a pushforward map between two measures; for example $P_{Y_{C,1}}=\mathrm{d}_{\#} P_{Y_{C,0}}$. The maps $h_j$ are the ``production functions'' linking the unobservable measures $\nu$ and $\nu^*$ to the potential outcomes. A dashed arrow indicates a map from a measure to itself. $P_{Y_{T,1}^\dagger}$ is the counterfactual outcome measure of the treated units had they not received treatment. {$\mathrm{d}$} is the natural trend map and {$\mathrm{T}$} is the map from an observed outcome to its counterfactual. The {observable data} is drawn from the four boxed measures.}
\label{fig:DID}
\end{figure}
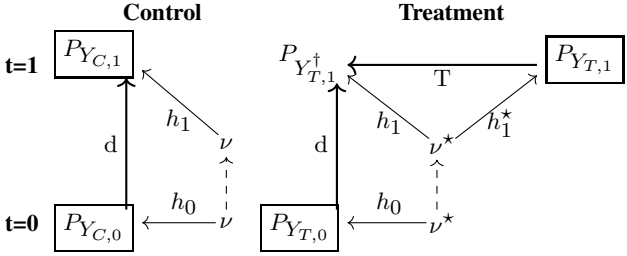

Each unit $i$ in an empirical setting is sampled from a larger population. Adapting the notation of \cite{athey2006identification}, a random variable $G_i$ denotes a unit’s treatment: $C$ for no treatment, $T$ for treatment. Let the random vectors $Y_{i;C,0}$ and $Y_{i;C,1}$ \revision{denote} unit $i$’s observable potential outcomes in the control group in the pre- and post-intervention periods, respectively; $Y_{i;T,0}$ and $Y_{i;T,1}$ are the unit’s observable potential outcomes when $G_i = T$. Each unit has indicator random variables $T_{i,0}$ and $T_{i,1}$ denoting whether an outcome was observed in each study period. 

The main assumptions needed for the model are as follows.
\begin{enumerate}
    \item Each potential outcome in the absence of treatment is generated by a deterministic production function $h(t,\cdot)$. That is, $Y_{i;C,0} = h(0,U_i)$, $Y_{i;T,0} = h(0,U_i^*)$, and $Y_{i;C,1} = h(1,U_i)$. 
    \item Moreover, the laws $\nu$ and $\nu^*$ of the random variables $U_i$ and $U_i^*$ describing the unobservable characteristics of the individual do not change over time \emph{within a treatment group}. In Figure \ref{fig:DID} this is captured by the fact that $\nu$ and $\nu^*$---which can be arbitrarily different---stay the same over time. 
\end{enumerate}

The basic idea is as in the linear case: the ``parallel trends'' assumption implies that the natural trend of the control group (the map ``$\mathrm{d}$'' in Figure \ref{fig:DID}) is the change in the outcome distribution of the treatment group had it not received treatment. By definition, for the control group, in both time periods, one observes the counterfactual outcome of no treatment (the two boxed measures $P_{Y_{C,0}}$ and \revision{$P_{Y_{C,1}}$} in the control group). For the treated group, one observes the counterfactual outcome of no treatment at $t=0$ and the potential outcome of being treated at $t=1$ (the two boxed probability measures $P_{Y_{T,0}}$ and $P_{Y_{T,1}}$ in the treatment group in Figure \ref{fig:DID}).

To isolate the effect of receiving the treatment between $t=0$ and $t=1$, we want to ``net the change'' over time in the observed outcomes in the treatment group by the change in the control group. 
In order to identify the counterfactual outcome of the treatment group had it not received treatment, we want to transplant the ``natural trend'' $\mathrm{d}$---that is the change in the outcome in the case where no treatment is administered---to the outcome $P_{Y_{T,0}}$ of the treatment group before treatment. 

There are \revision{infinitely many} ways to (i) define the natural trend and (ii) translate it to $P_{Y_{T,0}}$. 
A \revision{canonical} way in the univariate setting is to work with the order structure, which leads directly into the use of the monotone rearrangement. The additional assumptions needed for this are as follows \citep{athey2006identification}. 
\begin{enumerate}
\setcounter{enumi}{2}
\item The production functions $h(t,U)$ \revision{are assumed monotone in $U$, so that it coincides with the monotone rearrangement from the distribution of the unobservable to that of the corresponding potential outcome.}
\item \revision{The support of $U^*$ must be contained in the support of $U$, so that the natural trend $d = h_1 \circ h_0^{-1}$, which is defined on the image of $\mathrm{supp}(\nu)$, can be applied to the entire treatment distribution $P_{Y_{T,0}}$.}
\end{enumerate}

Under these assumptions---and if $U$ is either continuously or discretely distributed---the counterfactual CDF $F_{Y^\dagger_{T,1}}$ is identified as \citep[Theorem 3.1]{athey2006identification}
\[F_{Y^\dagger_{T,1}}(y) = F_{Y_{T,0}}\left(F^{-1}_{Y_{C,0}}\left(F_{Y_{C,1}}(y)\right)\right),\] and the difference between the observed $F_{Y_{T,1}(y)}$ and the induced $F_{Y^\dagger_{T,1}}(y)$ is the \emph{changes-in-changes estimator}.
Note that the monotone rearrangement in this identification result goes in \emph{the opposite direction} \revision{from what one might expect}, that is, from the post-treatment period to the pre-treatment period in the control group. The intuition is the same, however: we want to translate the natural trend from the control onto the treatment group to induce the counterfactual of what would have happened to the treatment group without treatment. 

The statistical estimator based on this argument was used in \cite{ropponen2011reconciling} to analyze the minimum wage debate as outlined in the introduction. By being able to estimate the entire counterfactual CDF and not just an average, \cite{ropponen2011reconciling} showed that there is a positive employment effect for smaller fast-food restaurants and a negative employment effect for larger ones. \revision{However, the univariate analysis cannot capture the dependence structure across outcome dimensions, which motivates the extension to multivariate outcomes.}

Finally, I want to point out that in the univariate setting, other approaches have been proposed which generalize the ``parallel trends'' assumption to allow for general heterogeneity. \revision{Other approaches have been proposed that replace monotonicity with alternative assumptions on the relationship between pre- and post-treatment distributions.} \cite{callaway2019quantile} directly assume that the copulas between $P_{Y_C,0}$ and $P_{Y_T,0}$ as well as $P_{Y_T,0}$ and $P_{Y_T,1}^\dagger$ are identical. \cite{roth2023parallel} analyze what happens when pointwise differences between the corresponding cumulative distribution functions are equal: $F_1(x) - F_0(x) = F_1^\dagger(x) - F_0^\star(x)$ for all $x \in \mathbb{R}$. \cite{bonhomme2011recovering} restrict the heterogeneity of the model to be additively separable and assume that the pointwise differences between the corresponding logarithms of the characteristic functions are equal.

\subsection{Extensions via (cyclic) comonotonicity}
One would think that an extension of the changes-in-changes idea to multivariate outcomes now follows the straightforward path: replace the monotone rearrangement by Brenier maps. While this is true to an extent, it is not quite so simple, surprisingly. The reason is that in the multivariate setting compositions of Brenier maps are not necessarily Brenier maps; unlike in the univariate setting, where this is true. The multivariate setting has recently been analyzed in detail in \cite{torous2024optimal}\revision{.}

 One key insight is that the monotonicity assumptions on the production functions in the univariate setting of \cite{athey2006identification} allow the latent variable to be entirely abstracted away\revision{.} This follows from the weaker notion of \emph{comonotonicity}. We say $h_0, h_1 : \mathbb{R}^d \to \mathbb{R}^d$ are \emph{comonotone} \citep{ekeland2012comonotonic} if
\[
\langle h_0(x) - h_0(y), h_1(x) - h_1(y) \rangle \geq 0
\]
for all $x, y \in \mathbb{R}^d$. Note that with an identity production function, $h(t,U) = U$, comonotonicity reduces to classical monotonicity. In the univariate case comonotonicity implies locally that the derivatives $h_0' \cdot h_1' \geq 0$ have the same sign; in addition to further global constraints. As a concrete example, if $h_0$ is a polynomial and $h_1$ its pointwise scaling by some $\gamma > 0$, then this pair of functions is comonotone because they have the same sign between all zeros and hence the same signed difference between any pair of points. This example emphasizes that $h_0$ and $h_1$ need not be individually monotone themselves for them to be comonotone. \revision{Cyclical comonotonicity is thus a weaker condition than requiring each production function to be individually monotone; however, it suffices for identification in the multivariate setting.}

We have now seen that the monotone production functions assumption in \cite{athey2006identification} in the univariate setting implies that the natural trend $d$ is cyclically monotone. Furthermore, the production functions are comonotone. To extend the changes-in-changes estimator to higher dimensions, \cite{torous2024optimal} exploit the idea of comonotonicity. They show that the assumption of \emph{cyclically comonotone production functions} has the desired properties to generalize the changes-in-changes estimator.
\begin{definition}[Cyclic comonotonicity, \citep{torous2024optimal}]
Two production functions $h_0$ and $h_1$ are cyclically comonotone if for any positive integer $m$ and any cycle \[u_1, \dots, u_m, u_{m+1} = u_1\] in their common domain, it holds
\[
\sum_{i=1}^{m} \langle h_0(u_i), h_1(u_i) - h_1(u_{i+1}) \rangle \geq 0.
\]
\end{definition}

Just as cyclical monotonicity collapses to monotonicity when $m = 2$, cyclical comonotonicity collapses to comonotonicity in that case. Whenever $h_0$ has an inverse $h_0^{-1}$, cyclical comonotonicity implies that the natural trend $d = h_1 \circ h_0^{-1}$ is also cyclically monotone \citep[Theorem 2]{torous2024optimal}. By the uniqueness of Brenier maps \citep[e.g.][Theorem 2.12]{villani2008optimal}, $d$ is then the unique Brenier map such that $P_{Y_C,1} = d\#P_{Y_C,0}$. For the identification result, one needs some structural assumptions on the supports, mirroring the univariate setting: the observable measures $P_{Y_C,t}$, $P_{Y_T,t}$, $t = 0,1$, and the counterfactual measure $P_{Y_T,1}^\dagger$ are supported on proper convex subsets $K_t$, $K_t^\star$, and $K_1^\dagger$ of $\mathbb{R}^d$ and are absolutely continuous with respect to Lebesgue measure. Moreover, $K_0^\star \subset K_0$, which is required to transport the map entirely to the new domain.

\begin{theorem}[Multivariate extension of the changes\--in\--changes estimator, \citep{torous2024optimal}]
Consider the causal model depicted in Figure \ref{fig:DID}. Let the above regularity assumptions hold. Moreover, assume that the production function $h_0$ has a well-defined inverse and that $h_0$ and $h_1$ are cyclically comonotone in the sense of (6). Then there exists a unique map $d : K_0^\star \to K_1^\dagger$. It is the Brenier map from $P_{Y_C,0}$ to $P_{Y_C,1}$. The counterfactual distribution $P_{Y_T,1}^\dagger$ of the treated unit had it not received treatment is then identified via
\[
P_{Y_T,1}^\dagger = d\#P_{Y_T,0}.
\]
\end{theorem}

\revision{As in the univariate case, when the measures are not absolutely continuous with respect to Lebesgue measure, the optimal transport map is no longer unique, and one can only obtain bounds on the counterfactual measures.} Moreover, the optimal transport approach allows us to identify the actual counterfactual random variable, a generalization of \citep{athey2006identification}. 

\begin{corollary}[\citep{torous2024optimal}]
Consider the setting and assumptions from the previous theorem. If the production function $h_1$ has a well-defined inverse and $h_1^\star$ and $h_1$ are cyclically comonotone, then there exists a unique map $T : K_1^\star \to K_1^\dagger$ such that
\[
P_{Y_T,1}^\dagger = T_\#P_{Y_T,1}.
\]
$T$ is the Brenier map from $P_{Y_T,1}$ to $P_{Y_T,1}^\dagger$. $Y_{T,1}^\dagger$ is then identified via
\[
Y_{T,1}^\dagger = T(Y_{T,1}).
\]
\end{corollary}

To show that the multivariate extension is useful, consider again the minimum wage setting. A multivariate outcome can now for instance also distinguish between full-time and part-time employees---two groups that a priori should be affected very differently by a minimum wage increase---while accounting for the correlation structure between those sets of employees. \cite{torous2024optimal} reproduce the same result of \cite{ropponen2011reconciling} about the size of the restaurants in question: positive effects on employment for smaller restaurants and negative effects for larger ones. Now, one can also split the problem into full-time and part-time employees. They find a strong (average) positive effect for full-time employees and a strong (average) negative effect for part-time employees, an additional dimension to further understand the effects of increasing the minimum wage. 

\subsection{Individual heterogeneity vs.~the entire system}
The above changes-in-changes estimator and its extension are focused on identification of individual heterogeneity. This is clearer in the univariate setting, where the monotone rearrangement implies that quantiles of $U$ are preserved when mapped to the potential outcomes. In the next section, when introducing a generalization of the synthetic control method, I focus on the systems view. There, the level of interest are not the individual quantiles, but the distribution as a whole. This also means that the standard monotone rearrangement will not underlie the identification strategy. \revision{A similar} view \revision{is powerful in the difference-in-differences setting, which is currently pursued by \cite{saidi2026parallel} via parallel transport in the Wasserstein space}.  This \revision{has} important applications in many areas of science. \revision{For instance, in} single-cell RNA-seq data\revision{,} measuring the gene expression levels of a cell is a destructive process, and as a result a given cell may be only measured once, fitting into the standard causal inference setting we consider \citep{bunne2024optimal}.

\section{Recreating trends: synthetic controls}
\revision{In this section, following the synthetic controls literature, $j$ indexes aggregate units, $t$ indexes time periods, $Y_{jt}$ denotes the outcome for unit $j$ at time $t$, and $X_j$ denotes observed covariates of unit $j$.}

The existence of parallel trends is the fundamental assumption that allows identification in difference-in-difference settings. If one has access to more pre-treatment time periods than just one, this assumption can be tested by checking if the pre-trends between treatment and control are sufficiently parallel. Very often, the parallel trends assumption is violated and \revision{cannot be maintained}. In such settings, it is still possible to obtain estimates of causal effects if one has access to several units that never receive treatment. The seminal idea, introduced in \cite{abadie2010synthetic, abadie2003economic} is to find a convex combination of such control units that can replicate the pre-treatment trend of the target unit as closely as possible. 

\subsection{Classic synthetic controls}
The classic method of synthetic controls is designed for \emph{comparative case studies} \citep{abadie2021using, abadie2015comparative}, where an \emph{aggregate} system or unit, such as a state or industry sector, is exposed to a treatment at one point in time $t=t^*$ and stays treated thereafter. This setting is a case study in that there are only few units to consider: in the most extreme case, only one treated unit and a few control units. Moreover, each unit is observed over several time periods $t=T_0,\ldots, T$, with $T_0<t^*<T$. I call $t \leq t^*$ the pre-treatment- or pre-treatment periods and $t > t^*$ the post-intervention- or post-treatment periods.

\revision{There is broad consensus} that in such settings, a combination of untreated aggregate units can provide a better control group than a single one \citep{abadie2021using}. The original method of synthetic controls then provides a canonical approach to obtain a close replication by projecting the outcome of the target (the treated unit) onto the convex hull of the potential control units to generate the synthetic control. 

Let $\{Y_{0t}\}_{t\in[T_0,T]}$ be the observed time series of the treated unit, and $\{Y_{jt}\}_{t\in[T_0,T]}$ for $j=1,\ldots, J$ be the observed time series of aggregate units that never receive treatment---the potential controls. The key quantity to estimate is $Y_{0t,N}$, the outcome of the treatment unit had it not received the treatment in the post-intervention periods. Based on this, one defines the effect $\alpha_{jt} = Y_{jt,I} - Y_{jt,N}$ of the intervention for unit $j$ at time $t$, so that one can write the observable outcome in terms of the counterfactual notation as
\[
Y_{jt} = Y_{jt,N} + \alpha_{jt} D_{jt},
\]
where $D_{jt}\in\{0,1\}$ is the treatment indicator. 

The goal in the suggested setting is to estimate the treatment effect on the treated group in the post-treatment period, i.e.
\[
\alpha_{0t} = Y_{0t,I} - Y_{0t,N} = Y_{0t} - Y_{0t,N} \quad \text{for } t > t^*,
\]
which requires a model for the unobservable $Y_{0t,N}$. \cite{abadie2010synthetic} introduce a linear factor model
\[
Y_{jt,N} = \delta_t + \theta_t X_j + w_t \mu_j + \varepsilon_{jt},
\]
where $\delta_t$ is a univariate factor, $w_t$ is a vector containing factors whose loadings are captured in $\mu_j$, and where $X_j$ are observed covariates of the respective units. The error terms $\varepsilon_{jt}$ are zero mean transitory shocks. The idea for the synthetic controls method is that $Y_{jt,N} = Y_{jt}$ for $t > t^*$ and $j = 1, \dots, J$, so that the treatment effect on the treatment group can be obtained by a weighted average
\[
\hat{\alpha}_{0t} = Y_{0t} - \sum_{j=1}^{J} \lambda_j^\star Y_{jt},
\]
where $\{\lambda_j^\star\}_{j=1,\dots,J}$ is an optimal set of weights.

The classical synthetic controls estimator in this setting then proceeds in two stages. In the first stage, one obtains the optimal weights $\lambda^\star := \{\lambda_j^\star\}_{j=1,\dots,J}$ which lie in the $J$-dimensional probability simplex $\Delta^J$ and are chosen such that they minimize a weighted Euclidean distance
\[\left( \sum_{k=1}^{K} v_k (X_{k0} - \lambda_1 X_{k1} - \dots - \lambda_{J} X_{kJ})^2 \right)^{1/2},
\]
where $v \in \Delta^K$ is another set of weights which needs to be chosen by the researcher. \cite{abadie2021using, abadie2010synthetic, abadie2015comparative, arkhangelsky2021synthetic} provide possible choices for $v$. In the second stage, the obtained optimal weights $\{\lambda_j^\star\}_{j=1,\dots,J}$ from this minimization are used to create $\hat{Y}_{0t}^N$ in the post-treatment periods as
\[
\hat{Y}_{0t}^N = \sum_{j=1}^{J} \lambda_j^\star Y_{jt}, \quad \text{for } t > t^*,
\]
based on which one can estimate $\hat{\alpha}_{0t} = Y_{0t} - \hat{Y}_{0t}^N$. Since the introduction of the original method, there have been \revision{numerous} extensions and generalizations, see the overview \cite{abadie2021using}. 

\subsection{Distributional synthetic controls}
So far, there is no direct connection between the synthetic controls method and optimal transport. Recently, however, \cite{gunsilius2023distributional} introduced an extension to the synthetic control method in settings where more information about the aggregate unit in question is available. This extension is based on optimal transportation, in particular barycenters in Wasserstein space \citep{agueh2011barycenters}. 

The setting is one where the researcher observes $j=0,\ldots,J$ aggregate systems/units over time periods $t$, but also has access to data $Y_{ijt}$ \emph{within the system}. A canonical example is a state $j$, for instance New Jersey, where the researcher observes fast-food restaurants $i$ within the state $j$ at time $t$. Then one can consider the system ``fast-food restaurants within the state of New Jersey''. The difference to classic longitudinal or panel data settings is that the level of interest of the causal effect is the state level $j$, \emph{not} the individual level $i$ for each restaurant within the state. In particular, this means that the method is applicable in settings where the individuals $i$ within the system $j$ cannot be traced over time. In the above example this can happen when restaurants close or new ones open for instance. Another setting is employees $i$ in a company $j$ and the treatment is administered at the company level. Such a setting has recently been analyzed in \cite{van2024return}. 

The method is designed for univariate outcomes, but the concept can be generalized to multivariate outcomes \citep{gunsilius2024tangential}, as I briefly show below. The quantity of interest is the quantile function $F_{Y_{jt}}^{-1}(q)$. 
The goal is to estimate the counterfactual quantile function $F_{Y_{0t,N}}^{-1}(q)$ of the treated unit had it not received treatment by an optimally weighted average of the control quantile functions $F_{Y_{jt}}^{-1}(q)$ for all $j = 1, \dots, J,$ $t > t^*$:
\[
F_{Y_{0t,N}}^{-1}(q) = \sum_{j=1}^{J} \lambda_j^\star F_{Y_{jt}}^{-1}(q) \quad \text{for all } q \in (0,1).
\]
The weights, as in the classical method, are obtained by trying to replicate the treatment quantile function $F^{-1}_{Y_{0t}}(q)$ as closely as possible by a convex combination of the control quantile functions. 

\revision{Since the $2$-Wasserstein space for measures supported on the real line is flat \citep{kloeckner2010geometric}---meaning that convex combinations of quantile functions are again valid quantile functions---one can exploit this linear structure directly.} The optimization for the replication becomes:
\[
\vec{\lambda}_t^* = \arg\min_{\vec{\lambda} \in \Delta^J} \int_0^1 \left| \sum_{j=1}^{J} \lambda_j F_{Y_{jt}}^{-1}(q) - F_{Y_{0t}}^{-1}(q) \right|^2 dq
\]
for all $t<t^*$. 

In some settings, for instance, when it is known that the distributions are mixtures, it is useful to work with distribution functions instead of quantiles. In this case, using the $1$-Wasserstein distance in CDF form is useful:
\[
\vec{\lambda}_t^* = \arg\min_{\vec{\lambda} \in \Delta^J} \int_{\mathbb{R}} \left| \sum_{j=1}^{J} \lambda_j F_{Y_{jt}}(y) - F_{Y_{0t}}(y) \right| dy.
\]
To obtain one set of weights $\lambda^*$ over all pre-treatment time periods $t$, one usually forms another weighted average over time; \cite{arkhangelsky2021synthetic} provide useful approaches. 
In most practical settings the quantile functions are not given directly but have to be estimated from observations $\{Y_{ijt}\}$, $i=1,\ldots,n_j$, $j=0,\ldots,J$, $t\in [T_0,T]$. 

\revision{In addition to the proposed estimator, one can ask what restrictions the identification strategy imposes on the causal model in this setting.} While the classic method is based on a linear factor model as seen above, the question is how general one can be for a causal model of the form 
\[
P_{Y_{jt,N}} = h_t\# P_{U_{jt}} \quad \text{for} \quad P_{U_{jt}} = g_t\# P_{U_{j(t-1)}}.
\]

\revision{Since the synthetic control method replicates across groups but extrapolates over time---the key identification assumption is that the weights $\vec{\lambda}^*$ obtained in the pre-treatment periods stay optimal in the post-treatment periods---the maps $h$ and $g$ must preserve the relative Wasserstein distances between units, i.e., they must be isometries between the respective measures \citep[Appendix]{gunsilius2023distributional}. 
\begin{assumption}[Weight stability]\label{ass:weight_stability}
The optimal weights $\vec{\lambda}^* \in \Delta^J$ obtained by replicating the treated unit's distribution in the pre-treatment periods $t < t^*$ remain optimal in the post-treatment periods $t \geq t^*$, i.e., $P_{Y_{0t,N}} = \sum_{j=1}^{J} \lambda_j^* P_{Y_{jt}}$ for all $t \geq t^*$.
\end{assumption}
Assumption~\ref{ass:weight_stability} implies that the maps $h$ and $g$ in the causal model must preserve the relative Wasserstein distances between units, i.e., they must be isometries in the $2$-Wasserstein space \citep[Appendix]{gunsilius2023distributional}.}

In the $2$-Wasserstein space on the line the set of such isometries is larger than the standard isometries on $\mathbb{R}$ \citep{kloeckner2010geometric}. \revision{The full class of $2$-Wasserstein isometries on the line is characterized in \cite{kloeckner2010geometric} and includes non-trivial maps beyond standard Euclidean isometries; \cite{gunsilius2023distributional} focuses on the tractable subclass of linear maps, i.e.\ $h(t, U_{jt}) = \alpha_t+\beta_t U_{jt}$. Since the maps must be isometries in the $2$-Wasserstein space and the only smooth isometries on the line are affine, linearity is effectively necessary; the only non-linear alternatives are the exotic isometries of \cite{kloeckner2010geometric}. This restriction extends to the classic method of \cite{abadie2010synthetic}.}

\revision{While the univariate case is the most useful in applied settings, the multivariate setting is also of interest. Unlike the univariate case, the $2$-Wasserstein space over $\mathbb{R}^d$ is non-negatively curved \citep{kloeckner2010geometric}, so that convex combinations of measures are no longer well-defined via simple weighted averages. One hence needs to replace the weighted average by a metric analogue: the barycenter \citep{agueh2011barycenters}.} A direct approach for obtaining the optimal weights in the univariate setting would be
\[
\vec{\lambda}_t^* = \underset{\lambda \in \Delta^J}{\mathrm{argmin}} W_2^2(P_{Y_{0t}}, P(\lambda))
\]

where 

\[
P(\vec{\lambda}) = \underset{{P \in \mathcal{P}_2(\mathbb{R}^d)}} {\mathrm{argmin}} \sum_{j=1}^{J} \frac{\lambda_j}{2} W_2^2(P, P_{Y_{jt}})
\]
is the barycenter in the $2$-Wasserstein space for the weights $\vec{\lambda} = (\lambda_1,\ldots, \lambda_J)\in \Delta^J$. This expression has been used in different settings, for instance in applications in computer vision \citep{bonneel2016wasserstein}. \revision{Solving this problem directly is challenging: it is a bilevel optimization problem in which the inner barycenter problem is non-convex in $\mathbb{R}^d$, and the outer problem may admit multiple local optima.}

One way to achieve this is via the tangent structure. \cite{gunsilius2024tangential} introduced a notion of tangential projection in the $2$-Wasserstein space that is efficient to implement via linear regression. The idea is to lift the problem to the tangent space centered at the target measure $P_{Y_{0t}}$, which linearizes the problem. \revision{Related approaches have been developed in computer vision and machine learning \citep{werenski2022measure, merigot2020quantitative, fan2023generating}. The tangential projection approach of \cite{gunsilius2024tangential} is significantly more efficient and applies to general measures that need not be absolutely continuous with respect to Lebesgue measure. Several open questions remain, including the large-sample distribution theory for the distributional synthetic controls estimator, which is currently being developed \citep{van2024return, zhang2024asymptotic}. A promising direction is to extend these connections beyond distributions to metric measure spaces, using the Gromov-Wasserstein distance \citep{memoli2014gromov}, which would allow comparisons between units whose outcome spaces have different geometries.}

\section{Matching and unbalanced optimal transport}
\revision{In this section, following the matching and treatment effects literature, $T\in\{0,1\}$ denotes the treatment indicator, $X\in\mathbb{R}^d$ the observed covariates, and $(Y_0,Y_1)$ the potential outcomes.}
\revision{Perhaps the most natural connection between optimal transport and causal inference is via matching \citep{stuart2010matching}.} In the most basic setup, we have two treatment groups $T\in\{0,1\}$, the corresponding potential outcomes $(Y_0,Y_1)$, and a (potentially high-dimensional) set of observable covariates $X\in\mathbb{R}^d$. The additional information supplied by the covariates in each group can be used in many different ways \citep{stuart2010matching}. One is \emph{balancing} \citep{ben2021balancing}, which essentially is a pre-processing method for downstream estimation of causal effects. In this step, one selects a sample where the treatment and control samples are more similar than in the original sample \citep{imbens2015causal}. Another main way in which the additional information of the covariates is used is by matching on them. 

Optimal transport is potentially useful for both. Optimal transportation is naturally suited to be used as a pre-processing step to balance covariates. Many related methods have been influential in this area \citep[e.g.][]{hainmueller2012entropy, imai2014covariate}, and classic optimal transportation is beginning to be used in related methods \citep[e.g.][]{arellano2023recovering, dong2023causal, yan2024reducing}. However, classic optimal transport is \emph{not} well-suited to be used in the second problem, that is, as a direct estimator for obtaining causal effects via matching. The issue here is the mass-preserving constraint of classic optimal transport, which requires every individual in treatment- and control group to be matched, as I now argue. 

The idea of using matching for obtaining estimates of causal effects is based on two key assumptions.
\revision{
\begin{assumption}[Unconfoundedness]\label{ass:unconfoundedness}
$(Y_0, Y_1) \perp T \mid X$, i.e., conditional on the observed covariates $X$, treatment assignment $T$ carries no additional information about the potential outcomes.
\end{assumption}
\begin{assumption}[Overlap]\label{ass:overlap}\mbox{}
\[\mathrm{supp}(P_{X|T=0}) = \mathrm{supp}(P_{X|T=1}),\] i.e., every value of the covariates that occurs in one treatment group also occurs in the other.
\end{assumption}
Under both assumptions, one can identify the average treatment effect by comparing individuals with the same observable covariates across treatment groups and averaging over these matched pairs. }Averaging over the differences in outcomes over all of these matched pairs should identify the average treatment effect.  

Due to random sampling, perfect matches between the groups are rare in practice. The question is then how to find good matches. There is a vast literature on matching approaches \citep{stuart2010matching}, including propensity score matching \citep{austin2011introduction, rosenbaum1983central}, and other methods of direct matching \citep[e.g.][and references therein]{kallus2020generalized, morucci2020adaptive, rosenbaum1989optimal, rosenbaum2020modern}. The most widely used idea is to specify a distance, e.g.~a weighted Euclidean distance or some discrepancy measure like the Kullback-Leibler divergence, and proceed via some form of $k$-nearest neighbor matching \citep{rubin1973matching, stuart2010matching}. In the most basic form, every individual in the treatment group gets matched with an individual in the control group. This will lead to biased estimates if the overlap of the supports of the covariate distributions of the two groups is not perfect in the population. In this case, \revision{individuals not in the intersection of the two supports do not have a good match in the other group and should not be matched for an unbiased estimate of the treatment effect.}

\revision{While Assumption~\ref{ass:overlap} is standard in this literature, it is often violated.} This implies that matching via classic optimal transportation is not the right approach: the measure-preserving constraint in the Monge-Kantorovich problem implies that all individuals in both groups have to be matched. Especially in finite samples this introduces excessive bias into the estimator.

The solution to this problem is to use unbalanced optimal transportation \citep{chizat2018unbalanced, liero2018optimal, sejourne2023unbalanced}, which relaxes the measure-preservation constraint, hence allowing for individuals to remain unmatched. The unbalanced optimal transport problem is 
\begin{equation}\label{eq:unbalanced}
\begin{aligned}
\inf_{\gamma\in\mathcal{M}^+\left(\mathcal{X}\times\mathcal{Y}\right)} \int_{\mathcal{X}\times\mathcal{Y}} c(x,y)d\gamma(x,y) + \varepsilon KL(\gamma || P_{X_0}\otimes P_{X_1}) \\+ \rho D_\phi\left(\pi_0\gamma\vert\vert P_{X_0}\right) + \rho D_\phi\left(\pi_1\gamma\vert\vert P_{X_1}\right),
\end{aligned}
\end{equation}
where $KL(\gamma || P_{X_0}\otimes P_{X_1})$ is the Kullback-Leibler divergence \citep{kullback1951information, kullback1959information} between the optimal Kantorovich coupling $\gamma$ and the independence coupling of the two marginals $P_{X_0}$ and $P_{X_1}$, $\pi_j\gamma$ denotes the projection onto the $j$-th marginal of the coupling $\gamma$, $D_\varphi$ denotes the $\phi$-divergence (or Csisz\`ar-divergence) \citep{csiszar1967information, csiszar1975divergence}, and $\mathcal{M}^+\left(\mathcal{X}\right)$ denotes the set of all non-negative finite measures on $\mathbb{R}^d\times\mathbb{R}^d$. The first penalty term using the Kullback-Leibler divergence is only used to improve computational properties, analogous to the classic Sinkhorn regularization \citep{cuturi2013sinkhorn, galichon2010matching}, and can be dropped by setting $\varepsilon=0$.  

\revision{The parameter $\rho>0$ governs how strictly the marginal constraints are enforced: large $\rho$ penalizes deviation from the original marginals $P_{X_0}$ and $P_{X_1}$ heavily, recovering near-balanced optimal transport, while small $\rho$ relaxes the marginal constraints and allows more individuals to remain unmatched.}

The main difference to the classical optimal transport problem are (i) that the optimal coupling $\gamma$ does not need to be a probability measure and (ii) the addition of the $\phi$-divergence terms, which allow for the creation and destruction of mass and do not enforce the constraint from classical optimal transport that the corresponding measures $P_{X_0}$ and $P_{X_1}$ need to have the same mass. In practice, this means that the optimal couplings obtained will provide partial optimal matches \citep{figalli2010optimal, kitagawa2015multi} where individuals that do not have a close enough match are automatically discarded. 

\revision{This automatic and disciplined way of matching two groups based on relative distance is expected to yield finite sample improvements in mean-squared error over existing methods, particularly when Assumption~\ref{ass:overlap} is violated.} Recently, \cite{gunsilius2021matching} analyzed the statistical properties of unbalanced optimal transport problems as the covariate measures $P_{X_0}$ and $P_{X_1}$ are approximated by empirical measures, showing that for fixed penalties $\varepsilon,\rho>0$, as the number of observations increase, the limit element of the empirical process is Gaussian. Moreover, they show that as $\rho\to 0$ after $\varepsilon\to0$ in the population problem \eqref{eq:unbalanced} for regular Csisz\`ar divergences, the optimal coupling $\gamma$ will only put weight on perfect matches from both groups, i.e., matches where the covariates coincide perfectly. This implies that, at least in the population, matching via unbalanced optimal transport can automatically provide unbiased estimates of the average treatment effect \emph{on the intersection of the supports of the two covariate distributions} as the balancing penalty $\rho$ vanishes. 

This is backed by simulations in \cite{gunsilius2021matching}, which show that an unbalanced approach can beat existing benchmarks in such settings. To fully introduce unbalanced optimal transportation into the toolbox of applied causal inference researchers, \emph{proof} of its superiority over other methods is needed. One way to do it is to analyze the statistical properties of \eqref{eq:unbalanced} in the setting where the penalty terms $\varepsilon$ and $\rho$ are \emph{data-dependent} and vanish at specific data-dependent rates. \revision{A natural conjecture is that by choosing appropriate divergence penalties and data-dependent rates of convergence for $\rho_n$, one can establish that unbalanced optimal transport achieves lower mean-squared error than standard matching estimators when Assumption~\ref{ass:overlap} is violated.}

\section{Conclusion}
This review introduced a selective overview of the uses of optimal transport theory in classic causal inference for observational data. The goal is to unify nomenclature and notation and to introduce new and potentially fruitful avenues for future research in this area. There are potentially many other uses of optimal transport in causal inference. However, the reviewed areas are particularly close since optimal transportation has built the foundation for most of the existing results in this area. It is my hope that this review simplifies the exchange between the two areas of research, and uncovers the connections and potential gains for both fields, on which future collaborations can be built.

\bibliographystyle{imsart-number} 
\bibliography{bibliography}       

@book{villani2008optimal,
  title={Optimal transport: old and new},
  author={Villani, C{\'e}dric},
  volume={338},
  year={2009},
  publisher={Springer Science \& Business Media}
}

@article{dawid1979conditional,
  title={Conditional independence in statistical theory},
  author={Dawid, A Philip},
  journal={Journal of the Royal Statistical Society Series B: Statistical Methodology},
  volume={41},
  number={1},
  pages={1--15},
  year={1979},
  publisher={Oxford University Press}
}

@article{caffarelli1990interior,
    author={Caffarelli, Luis A.},
    title={Interior {$W^{2,p}$} estimates for solutions of the {Monge--Amp{\`e}re} equation},
    journal={Annals of Mathematics},
    year={1990},
    volume={131},
    number={1},
    pages={135--150},
    doi={10.2307/1971510}
  }

@article{caffarelli1992regularity,
    author={Caffarelli, Luis A.},
    title={The regularity of mappings with a convex potential},
    journal={Journal of the American Mathematical Society},
    year={1992},
    volume={5},
    number={1},
    pages={99--104},
    doi={10.2307/2152752}
  }

@article{caffarelli1996boundary,
    author={Caffarelli, Luis A.},
    title={Boundary regularity of maps with convex potentials---{II}},
    journal={Annals of Mathematics},
    year={1996},
    volume={144},
    number={3},
    pages={453--496},
    doi={10.2307/2118564}
    }

@unpublished{gunsilius2025schrodinger,
  title={Partial identification with {S}chr{\"o}dinger bridges},
  author={Gunsilius, Florian F and {Nunes Costa}, Bruno},
  year={2026},
  note={Working paper, Emory University and University of Michigan}
}

@unpublished{saidi2026parallel,
  title={Parallel Transport on the Space of Probability Measures for Counterfactual Dynamics Prediction},
  author={Saidi, Tristan Luca and Mena, Gonzalo and Wasserman, Larry and Gunsilius, Florian},
  note={Working paper},
  year={2026}
}

@article{gentil2020entropic,
  title={An entropic interpolation proof of the HWI inequality},
  author={Gentil, Ivan and L{\'e}onard, Christian and Ripani, Luigia and Tamanini, Luca},
  journal={Stochastic Processes and their Applications},
  volume={130},
  number={2},
  pages={907--923},
  year={2020},
  publisher={Elsevier}
}

@article{callaway2019quantile,
  title={Quantile treatment effects in difference in differences models with panel data},
  author={Callaway, Brantly and Li, Tong},
  journal={Quantitative Economics},
  volume={10},
  number={4},
  pages={1579--1618},
  year={2019},
  publisher={Wiley Online Library}
}

@article{bunne2024optimal,
  title={Optimal transport for single-cell and spatial omics},
  author={Bunne, Charlotte and Schiebinger, Geoffrey and Krause, Andreas and Regev, Aviv and Cuturi, Marco},
  journal={Nature Reviews Methods Primers},
  volume={4},
  number={1},
  pages={58},
  year={2024},
  publisher={Nature Publishing Group UK London}
}

@article{abadie2003economic,
  title={The economic costs of conflict: A case study of the Basque Country},
  author={Abadie, Alberto and Gardeazabal, Javier},
  journal={American economic review},
  volume={93},
  number={1},
  pages={113--132},
  year={2003},
  publisher={American Economic Association}
}

@article{arellano2023recovering,
  title={Recovering latent variables by matching},
  author={Arellano, Manuel and Bonhomme, St{\'e}phane},
  journal={Journal of the American Statistical Association},
  volume={118},
  number={541},
  pages={693--706},
  year={2023},
  publisher={Taylor \& Francis}
}

@article{abadie2015comparative,
  title={Comparative politics and the synthetic control method},
  author={Abadie, Alberto and Diamond, Alexis and Hainmueller, Jens},
  journal={American Journal of Political Science},
  volume={59},
  number={2},
  pages={495--510},
  year={2015},
  publisher={Wiley Online Library}
}

@article{arkhangelsky2021synthetic,
  title={Synthetic difference-in-differences},
  author={Arkhangelsky, Dmitry and Athey, Susan and Hirshberg, David A and Imbens, Guido W and Wager, Stefan},
  journal={American Economic Review},
  volume={111},
  number={12},
  pages={4088--4118},
  year={2021},
  publisher={American Economic Association 2014 Broadway, Suite 305, Nashville, TN 37203}
}

@article{abadie2021using,
  title={Using synthetic controls: Feasibility, data requirements, and methodological aspects},
  author={Abadie, Alberto},
  journal={Journal of economic literature},
  volume={59},
  number={2},
  pages={391--425},
  year={2021},
  publisher={American Economic Association 2014 Broadway, Suite 305, Nashville, TN 37203-2425}
}

@article{abadie2010synthetic,
  title={Synthetic control methods for comparative case studies: Estimating the effect of California’s tobacco control program},
  author={Abadie, Alberto and Diamond, Alexis and Hainmueller, Jens},
  journal={Journal of the American statistical Association},
  volume={105},
  number={490},
  pages={493--505},
  year={2010},
  publisher={Taylor \& Francis}
}

@article{bonhomme2011recovering,
  title={Recovering distributions in difference-in-differences models: A comparison of selective and comprehensive schooling},
  author={Bonhomme, St{\'e}phane and Sauder, Ulrich},
  journal={Review of Economics and Statistics},
  volume={93},
  number={2},
  pages={479--494},
  year={2011},
  publisher={The MIT Press}
}

@article{roth2023parallel,
  title={When is parallel trends sensitive to functional form?},
  author={Roth, Jonathan and Sant'Anna, Pedro HC},
  journal={Econometrica},
  volume={91},
  number={2},
  pages={737--747},
  year={2023},
  publisher={Wiley Online Library}
}

@article{ekeland2012comonotonic,
  title={Comonotonic measures of multivariate risks},
  author={Ekeland, Ivar and Galichon, Alfred and Henry, Marc},
  journal={Mathematical Finance: An International Journal of Mathematics, Statistics and Financial Economics},
  volume={22},
  number={1},
  pages={109--132},
  year={2012},
  publisher={Wiley Online Library}
}

@article{blanchet2019quantifying,
  title={Quantifying distributional model risk via optimal transport},
  author={Blanchet, Jose and Murthy, Karthyek},
  journal={Mathematics of Operations Research},
  volume={44},
  number={2},
  pages={565--600},
  year={2019},
  publisher={INFORMS}
}

@article{blanchet2019robust,
  title={Robust Wasserstein profile inference and applications to machine learning},
  author={Blanchet, Jose and Kang, Yang and Murthy, Karthyek},
  journal={Journal of Applied Probability},
  volume={56},
  number={3},
  pages={830--857},
  year={2019},
  publisher={Cambridge University Press}
}

@inproceedings{tu2022optimal,
  title={Optimal transport for causal discovery},
  author={Tu, Ruibo and Zhang, Kun and Kjellstr{\"o}m, Hedvig and Zhang, Cheng},
  booktitle={ICLR 2022-The Tenth International Conference on Learning Representations (Virtual), Apr 25th-29th, 2022},
  year={2022},
  organization={International Conference on Learning Representations, ICLR}
}

@article{lin2023causal,
  title={Causal inference on distribution functions},
  author={Lin, Zhenhua and Kong, Dehan and Wang, Linbo},
  journal={Journal of the Royal Statistical Society Series B: Statistical Methodology},
  volume={85},
  number={2},
  pages={378--398},
  year={2023},
  publisher={Oxford University Press US}
}

@book{galichon2018optimal,
  title={Optimal transport methods in economics},
  author={Galichon, Alfred},
  year={2018},
  publisher={Princeton University Press}
}

@article{torous2024optimal,
  title={An optimal transport approach to estimating causal effects via nonlinear difference-in-differences},
  author={Torous, William and Gunsilius, Florian and Rigollet, Philippe},
  journal={Journal of Causal Inference},
  volume={12},
  number={1},
  pages={20230004},
  year={2024},
  publisher={De Gruyter}
}

@article{roth2023s,
  title={What’s trending in difference-in-differences? A synthesis of the recent econometrics literature},
  author={Roth, Jonathan and Sant’Anna, Pedro HC and Bilinski, Alyssa and Poe, John},
  journal={Journal of Econometrics},
  volume={235},
  number={2},
  pages={2218--2244},
  year={2023},
  publisher={Elsevier}
}

@article{abadie2005semiparametric,
  title={Semiparametric difference-in-differences estimators},
  author={Abadie, Alberto},
  journal={The review of economic studies},
  volume={72},
  number={1},
  pages={1--19},
  year={2005},
  publisher={Wiley-Blackwell}
}

@article{gao2023distributionally,
  title={Distributionally robust stochastic optimization with Wasserstein distance},
  author={Gao, Rui and Kleywegt, Anton},
  journal={Mathematics of Operations Research},
  volume={48},
  number={2},
  pages={603--655},
  year={2023},
  publisher={INFORMS}
}

@article{qu2024distributionally,
  title={Distributionally Robust Instrumental Variables Estimation},
  author={Qu, Zhaonan and Kwon, Yongchan},
  journal={arXiv preprint arXiv:2410.15634},
  year={2024}
}

@article{gunsilius2023independent,
  title={Independent nonlinear component analysis},
  author={Gunsilius, Florian and Schennach, Susanne},
  journal={Journal of the American Statistical Association},
  volume={118},
  number={542},
  pages={1305--1318},
  year={2023},
  publisher={Taylor \& Francis}
}

@article{dawsont1987large,
  title={Large deviations from the McKean-Vlasov limit for weakly interacting diffusions},
  author={Dawson, Donald A and G{\"a}rtner, J{\"u}rgen},
  journal={Stochastics: An International Journal of Probability and Stochastic Processes},
  volume={20},
  number={4},
  pages={247--308},
  year={1987},
  publisher={Taylor \& Francis}
}

@article{van2024return,
  title={Return to office and the tenure distribution},
  author={Van Dijcke, David and Gunsilius, Florian and Wright, Austin},
  journal={arXiv preprint arXiv:2405.04352},
  year={2024}
}

@article{bonneel2016wasserstein,
  title={Wasserstein barycentric coordinates: histogram regression using optimal transport.},
  author={Bonneel, Nicolas and Peyr{\'e}, Gabriel and Cuturi, Marco},
  journal={ACM Trans. Graph.},
  volume={35},
  number={4},
  pages={71--1},
  year={2016}
}

@inproceedings{werenski2022measure,
  title={Measure estimation in the barycentric coding model},
  author={Werenski, Matthew and Jiang, Ruijie and Tasissa, Abiy and Aeron, Shuchin and Murphy, James M},
  booktitle={International Conference on Machine Learning},
  pages={23781--23803},
  year={2022},
  organization={PMLR}
}

@inproceedings{merigot2020quantitative,
  title={Quantitative stability of optimal transport maps and linearization of the 2-Wasserstein space},
  author={M{\'e}rigot, Quentin and Delalande, Alex and Chazal, Frederic},
  booktitle={International Conference on Artificial Intelligence and Statistics},
  pages={3186--3196},
  year={2020},
  organization={PMLR}
}

@article{heckman1997making,
  title={Making the most out of programme evaluations and social experiments: Accounting for heterogeneity in programme impacts},
  author={Heckman, James J and Smith, Jeffrey and Clements, Nancy},
  journal={The Review of Economic Studies},
  volume={64},
  number={4},
  pages={487--535},
  year={1997},
  publisher={Wiley-Blackwell}
}

@article{kallus2020generalized,
  title={Generalized optimal matching methods for causal inference},
  author={Kallus, Nathan},
  journal={Journal of Machine Learning Research},
  volume={21},
  pages={1--54},
  year={2020}
}

@article{hainmueller2012entropy,
  title={Entropy balancing for causal effects: A multivariate reweighting method to produce balanced samples in observational studies},
  author={Hainmueller, Jens},
  journal={Political analysis},
  volume={20},
  number={1},
  pages={25--46},
  year={2012},
  publisher={Cambridge University Press}
}

@article{imai2014covariate,
  title={Covariate balancing propensity score},
  author={Imai, Kosuke and Ratkovic, Marc},
  journal={Journal of the Royal Statistical Society Series B: Statistical Methodology},
  volume={76},
  number={1},
  pages={243--263},
  year={2014},
  publisher={Oxford University Press}
}

@article{rosenbaum1983central,
  title={The central role of the propensity score in observational studies for causal effects},
  author={Rosenbaum, Paul R and Rubin, Donald B},
  journal={Biometrika},
  volume={70},
  number={1},
  pages={41--55},
  year={1983},
  publisher={Oxford University Press}
}

@inproceedings{morucci2020adaptive,
  title={Adaptive Hyper-box Matching for Interpretable Individualized Treatment Effect Estimation},
  author={Morucci, Marco and Orlandi, Vittorio and Roy, Sudeepa and Rudin, Cynthia and Volfovsky, Alexander},
  booktitle={Conference on Uncertainty in Artificial Intelligence},
  pages={1089--1098},
  year={2020},
  organization={PMLR}
}

@article{rosenbaum1989optimal,
  title={Optimal matching for observational studies},
  author={Rosenbaum, Paul R},
  journal={Journal of the American Statistical Association},
  volume={84},
  number={408},
  pages={1024--1032},
  year={1989},
  publisher={Taylor \& Francis}
}

@article{rosenbaum2020modern,
  title={Modern algorithms for matching in observational studies},
  author={Rosenbaum, Paul R},
  journal={Annual Review of Statistics and Its Application},
  volume={7},
  pages={143--176},
  year={2020},
  publisher={Annual Reviews}
}

@article{austin2011introduction,
  title={An introduction to propensity score methods for reducing the effects of confounding in observational studies},
  author={Austin, Peter C},
  journal={Multivariate {B}ehavioral {R}esearch},
  volume={46},
  number={3},
  pages={399--424},
  year={2011},
  publisher={Taylor \& Francis}
}

@article{rubin1973matching,
  title={Matching to remove bias in observational studies},
  author={Rubin, Donald B},
  journal={Biometrics},
  volume={29},
  number={1},
  pages={159--183},
  year={1973},
  publisher={JSTOR}
}

@article{kullback1951information,
  title={On information and sufficiency},
  author={Kullback, Solomon and Leibler, Richard A},
  journal={The {A}nnals of {M}athematical {S}tatistics},
  volume={22},
  number={1},
  pages={79--86},
  year={1951},
  publisher={JSTOR}
}

@book{kullback1959information,
  title={Information theory and statistics},
  author={Kullback, Solomon},
  year={1959},
  publisher={Wiley}
}

@inproceedings{kitagawa2015multi,
  title={The multi-marginal optimal partial transport problem},
  author={Kitagawa, Jun and Pass, Brendan},
  booktitle={Forum of Mathematics, Sigma},
  volume={3},
  year={2015},
  organization={Cambridge University Press}
}

@article{figalli2010optimal,
  title={The optimal partial transport problem},
  author={Figalli, Alessio},
  journal={Archive for rational mechanics and analysis},
  volume={195},
  number={2},
  pages={533--560},
  year={2010},
  publisher={Springer}
}

@article{csiszar1967information,
  title={Information-type measures of difference of probability distributions and indirect observation},
  author={Csisz{\'a}r, Imre},
  journal={{S}tudia {S}cientiarum Mathematicarum Hungarica},
  volume={2},
  pages={229--318},
  year={1967}
}

@article{csiszar1975divergence,
  title={I-divergence geometry of probability distributions and minimization problems},
  author={Csisz{\'a}r, Imre},
  journal={The annals of probability},
  pages={146--158},
  year={1975},
  publisher={JSTOR}
}

@misc{galichon2010matching,
  title={Matching with trade-offs: Revealed preferences over competing characteristics},
  author={Galichon, Alfred and Salani{\'e}, Bernard},
  year={2010},
  note={CEPR Discussion Paper No. DP7858}
}

@inproceedings{yan2024reducing,
  title={Reducing balancing error for causal inference via optimal transport},
  author={Yan, Yuguang and Zhou, Hao and Yang, Zeqin and Chen, Weilin and Cai, Ruichu and Hao, Zhifeng},
  booktitle={Proceedings of the 41st International Conference on Machine Learning},
  pages={55913--55927},
  year={2024}
}

@article{stuart2010matching,
  title={Matching methods for causal inference: A review and a look forward},
  author={Stuart, Elizabeth A},
  journal={Statistical science: a review journal of the Institute of Mathematical Statistics},
  volume={25},
  number={1},
  pages={1},
  year={2010}
}

@article{ben2021balancing,
  title={The balancing act in causal inference},
  author={Ben-Michael, Eli and Feller, Avi and Hirshberg, David A and Zubizarreta, Jos{\'e} R},
  journal={arXiv preprint arXiv:2110.14831},
  year={2021}
}

@article{klatt2022limit,
  title={Limit laws for empirical optimal solutions in random linear programs},
  author={Klatt, Marcel and Munk, Axel and Zemel, Yoav},
  journal={Annals of Operations Research},
  volume={315},
  number={1},
  pages={251--278},
  year={2022},
  publisher={Springer}
}

@article{fang2023inference,
  title={Inference for Large-Scale Linear Systems With Known Coefficients},
  author={Fang, Zheng and Santos, Andres and Shaikh, Azeem M and Torgovitsky, Alexander},
  journal={Econometrica},
  volume={91},
  number={1},
  pages={299--327},
  year={2023},
  publisher={Wiley Online Library}
}

@article{memoli2014gromov,
  title={The Gromov--Wasserstein distance: A brief overview},
  author={M{\'e}moli, Facundo},
  journal={Axioms},
  volume={3},
  number={3},
  pages={335--341},
  year={2014},
  publisher={MDPI}
}

@article{gunsilius2021matching,
  title={Matching for causal effects via multimarginal unbalanced optimal transport},
  author={Gunsilius, Florian and Xu, Yuliang},
  journal={arXiv preprint arXiv:2112.04398},
  year={2021}
}

@article{dong2023causal,
  title={Causal identification of single-cell experimental perturbation effects with CINEMA-OT},
  author={Dong, Mingze and Wang, Bao and Wei, Jessica and de O. Fonseca, Antonio H and Perry, Curtis J and Frey, Alexander and Ouerghi, Feriel and Foxman, Ellen F and Ishizuka, Jeffrey J and Dhodapkar, Rahul M},
  journal={Nature methods},
  volume={20},
  number={11},
  pages={1769--1779},
  year={2023},
  publisher={Nature Publishing Group US New York}
}

@article{zhang2024asymptotic,
  title={Asymptotic Properties of the Distributional Synthetic Controls},
  author={Zhang, Lu and Zhang, Xiaomeng and Zhang, Xinyu},
  journal={arXiv preprint arXiv:2405.00953},
  year={2024}
}

@inproceedings{fan2023generating,
  title={Generating synthetic datasets by interpolating along generalized geodesics},
  author={Fan, Jiaojiao and Alvarez-Melis, David},
  booktitle={Uncertainty in Artificial Intelligence},
  pages={571--581},
  year={2023},
  organization={PMLR}
}

@article{kloeckner2010geometric,
  title={A geometric study of Wasserstein spaces: Euclidean spaces},
  author={Kloeckner, Beno{\^\i}t},
  journal={Annali della Scuola Normale Superiore di Pisa-Classe di Scienze},
  volume={9},
  number={2},
  pages={297--323},
  year={2010}
}

@article{gunsilius2024tangential,
  title={Tangential wasserstein projections},
  author={Gunsilius, Florian and Hsieh, Meng Hsuan and Lee, Myung Jin},
  journal={Journal of Machine Learning Research},
  volume={25},
  number={69},
  pages={1--41},
  year={2024}
}

@article{follmer1988random,
  title={Random fields and diffusion processes},
  author={F{\"o}llmer, Hans},
  journal={Lect. Notes Math},
  volume={1362},
  pages={101--204},
  year={1988}
}

@inproceedings{cattiaux1995large,
  title={Large deviations and Nelson processes},
  author={Cattiaux, Patrick and Leonard, Christian},
  booktitle={Forum Math},
  volume={7},
  pages={95--115},
  year={1995}
}

@article{kilbertus2020class,
  title={A class of algorithms for general instrumental variable models},
  author={Kilbertus, Niki and Kusner, Matt J and Silva, Ricardo},
  journal={Advances in Neural Information Processing Systems},
  volume={33},
  pages={20108--20119},
  year={2020}
}

@book{bogachev2007measure,
  title={Measure theory},
  author={Bogachev, Vladimir Igorevich},
  volume={2},
  year={2007},
  publisher={Springer}
}

@article{cheridito2023optimal,
  title={Optimal transport and Wasserstein distances for causal models},
  author={Cheridito, Patrick and Eckstein, Stephan},
  journal={arXiv preprint arXiv:2303.14085},
  year={2023}
}

@article{backhoff2022stability,
  title={Stability of martingale optimal transport and weak optimal transport},
  author={Backhoff-Veraguas, Julio and Pammer, Gudmund},
  journal={The Annals of Applied Probability},
  volume={32},
  number={1},
  pages={721--752},
  year={2022},
  publisher={Institute of Mathematical Statistics}
}

@book{manski1999identification,
  title={Identification problems in the social sciences},
  author={Manski, Charles F},
  year={1999},
  publisher={Harvard University Press}
}

@article{gunsilius2019path,
  title={A path-sampling method to partially identify causal effects in instrumental variable models},
  author={Gunsilius, Florian},
  journal={arXiv preprint arXiv:1910.09502},
  year={2019}
}

@book{manski2003partial,
  title={Partial identification of probability distributions},
  author={Manski, Charles F},
  year={2003},
  publisher={Springer Science \& Business Media}
}

@article{russell2021sharp,
  title={Sharp bounds on functionals of the joint distribution in the analysis of treatment effects},
  author={Russell, Thomas M},
  journal={Journal of Business \& Economic Statistics},
  volume={39},
  number={2},
  pages={532--546},
  year={2021},
  publisher={Taylor \& Francis}
}

@article{kitagawa2021identification,
  title={The identification region of the potential outcome distributions under instrument independence},
  author={Kitagawa, Toru},
  journal={Journal of Econometrics},
  volume={225},
  number={2},
  pages={231--253},
  year={2021},
  publisher={Elsevier}
}

@article{manski1990nonparametric,
  title={Nonparametric bounds on treatment effects},
  author={Manski, Charles F},
  journal={The American Economic Review},
  volume={80},
  number={2},
  pages={319--323},
  year={1990},
  publisher={JSTOR}
}

@book{mascolell1995micro,
  author = {Mas-Colell, Andreu and Whinston, Michael D. and Green, Jerry R.},
  publisher = {Oxford University Press},
  title = {Microeconomic Theory},
  year = {1995}
}

@article{hoderlein2017corrigendum,
  title={Corrigendum: Instrumental variables with unrestricted heterogeneity and continuous treatment},
  author={Hoderlein, Stefan and Holzmann, Hajo and Kasy, Maximilian and Meister, Alexander},
  journal={The Review of Economic Studies},
  volume={84},
  number={2},
  pages={964--968},
  year={2017},
  publisher={Oxford University Press}
}

@book{wright1928tariff,
  title={The tariff on animal and vegetable oils},
  author={Wright, Philip Green},
  number={26},
  year={1928},
  publisher={Macmillan}
}

@article{tinbergen1930determination, 
    title = {Determination and interpretation of supply curves: an example},
    author={Jan Tinbergen}, 
    journal = {Zeitschrift f\"ur National\"okonomie},
    volume={1},
    number={5},
    pages={669–679},
    year={1930}
}

@article{haavelmo1943statistical,
 author = {Trygve Haavelmo},
 journal = {Econometrica},
 number = {1},
 pages = {1--12},
 publisher = {[Wiley, Econometric Society]},
 title = {The Statistical Implications of a System of Simultaneous Equations},
 volume = {11},
 year = {1943}
}

@article{holland1986statistics,
  title={Statistics and causal inference},
  author={Holland, Paul W},
  journal={Journal of the American statistical Association},
  volume={81},
  number={396},
  pages={945--960},
  year={1986},
  publisher={Taylor \& Francis}
}

@article{card1994minimum,
  title={Minimum Wages and Employment: A Case Study of the Fast-Food Industry in New Jersey and Pennsylvania},
  author={Card, David and Krueger, Alan B},
  journal={The American Economic Review},
  volume={84},
  number={4},
  pages={772--793},
  year={1994}
}

@article{neumark2000minimum,
  title={Minimum wages and employment: A case study of the fast-food industry in New Jersey and Pennsylvania: Comment},
  author={Neumark, David and Wascher, William},
  journal={American Economic Review},
  volume={90},
  number={5},
  pages={1362--1396},
  year={2000},
  publisher={American Economic Association}
}

@article{ropponen2011reconciling,
  title={Reconciling the evidence of Card and Krueger (1994) and Neumark and Wascher (2000)},
  author={Ropponen, Olli},
  journal={Journal of Applied Econometrics},
  volume={26},
  number={6},
  pages={1051--1057},
  year={2011},
  publisher={Wiley Online Library}
}

@article{monge1781memoire,
  title={M{\'e}moire sur la th{\'e}orie des d{\'e}blais et des remblais},
  author={Monge, Gaspard},
  journal={Mem. Math. Phys. Acad. Royale Sci.},
  pages={666--704},
  year={1781}
}

@article{kantorovich1942translocation,
  title={On the translocation of masses, CR Dokl},
  author={Kantorovich, LV},
  journal={Acad. Sci. URSS},
  volume={37},
  pages={191--201},
  year={1942}
}

@book{pearl2009causality,
  title={Causality},
  author={Pearl, Judea},
  year={2009},
  publisher={Cambridge university press}
}

@book{peters2017elements,
  title={Elements of causal inference: foundations and learning algorithms},
  author={Peters, Jonas and Janzing, Dominik and Sch{\"o}lkopf, Bernhard},
  year={2017},
  publisher={The MIT Press}
}

@article{imbens2020essay,
Author = {Imbens, Guido W.},
Title = {Potential Outcome and Directed Acyclic Graph Approaches to Causality: Relevance for Empirical Practice in Economics},
Journal = {Journal of Economic Literature},
Volume = {58},
Number = {4},
Year = {2020},
Month = {December},
Pages = {1129–79}
}

@article{doksum1974empirical,
  title={Empirical probability plots and statistical inference for nonlinear models in the two-sample case},
  author={Doksum, Kjell},
  journal={The annals of statistics},
  pages={267--277},
  year={1974},
  publisher={JSTOR}
}

@article{matzkin2003nonparametric,
  title={Nonparametric estimation of nonadditive random functions},
  author={Matzkin, Rosa L},
  journal={Econometrica},
  volume={71},
  number={5},
  pages={1339--1375},
  year={2003},
  publisher={Wiley Online Library}
}

@book{imbens2015causal,
  title={Causal inference in statistics, social, and biomedical sciences},
  author={Imbens, Guido W and Rubin, Donald B},
  year={2015},
  publisher={Cambridge university press}
}

@book{angrist2009mostly,
  title={Mostly harmless econometrics: An empiricist's companion},
  author={Angrist, Joshua D and Pischke, J{\"o}rn-Steffen},
  year={2009},
  publisher={Princeton university press}
}

@article{heckman1999pre,
  title={The pre-programme earnings dip and the determinants of participation in a social programme. Implications for simple programme evaluation strategies},
  author={Heckman, James J and Smith, Jeffrey A},
  journal={The Economic Journal},
  volume={109},
  number={457},
  pages={313--348},
  year={1999},
  publisher={Wiley Online Library}
}

@article{neyman1990application,
  title={On the application of probability theory to agricultural experiments. Essay on principles. Section 9.},
  author={Neyman, Jerzy},
  journal={Statistical Science},
  pages={465--472},
  year={1990},
  publisher={JSTOR}
}

@article{imbens2009identification,
  title={Identification and estimation of triangular simultaneous equations models without additivity},
  author={Imbens, Guido W and Newey, Whitney K},
  journal={Econometrica},
  volume={77},
  number={5},
  pages={1481--1512},
  year={2009},
  publisher={Wiley Online Library}
}

@article{heckman2005structural,
  title={Structural equations, treatment effects, and econometric policy evaluation 1},
  author={Heckman, James J and Vytlacil, Edward},
  journal={Econometrica},
  volume={73},
  number={3},
  pages={669--738},
  year={2005},
  publisher={Wiley Online Library}
}

@article{torgovitsky2015identification,
  title={Identification of nonseparable models using instruments with small support},
  author={Torgovitsky, Alexander},
  journal={Econometrica},
  volume={83},
  number={3},
  pages={1185--1197},
  year={2015},
  publisher={Wiley Online Library}
}

@article{athey2006identification,
  title={Identification and inference in nonlinear difference-in-differences models},
  author={Athey, Susan and Imbens, Guido W},
  journal={Econometrica},
  volume={74},
  number={2},
  pages={431--497},
  year={2006},
  publisher={Wiley Online Library}
}

@article{benamou2000computational,
  title={A computational fluid mechanics solution to the Monge-Kantorovich mass transfer problem},
  author={Benamou, Jean-David and Brenier, Yann},
  journal={Numerische Mathematik},
  volume={84},
  number={3},
  pages={375--393},
  year={2000},
  publisher={Springer-Verlag Berlin/Heidelberg}
}

@article{backhoff2022applications,
  title={Applications of weak transport theory},
  author={Backhoff-Veraguas, Julio and Pammer, Gudmund},
  journal={Bernoulli},
  volume={28},
  number={1},
  pages={370--394},
  year={2022},
  publisher={Bernoulli Society for Mathematical Statistics and Probability}
}

@article{gozlan2017kantorovich,
  title={Kantorovich duality for general transport costs and applications},
  author={Gozlan, Nathael and Roberto, Cyril and Samson, Paul-Marie and Tetali, Prasad},
  journal={Journal of Functional Analysis},
  volume={273},
  number={11},
  pages={3327--3405},
  year={2017},
  publisher={Elsevier}
}

@article{agueh2011barycenters,
  title={Barycenters in the Wasserstein space},
  author={Agueh, Martial and Carlier, Guillaume},
  journal={SIAM Journal on Mathematical Analysis},
  volume={43},
  number={2},
  pages={904--924},
  year={2011},
  publisher={SIAM}
}

@article{carlier2010matching,
  title={Matching for teams},
  author={Carlier, Guillaume and Ekeland, Ivar},
  journal={Economic theory},
  volume={42},
  pages={397--418},
  year={2010},
  publisher={Springer}
}

@article{cuturi2013sinkhorn,
  title={Sinkhorn distances: Lightspeed computation of optimal transport},
  author={Cuturi, Marco},
  journal={Advances in neural information processing systems},
  volume={26},
  year={2013}
}

@article{chizat2018unbalanced,
  title={Unbalanced optimal transport: Dynamic and Kantorovich formulations},
  author={Chizat, Lenaic and Peyr{\'e}, Gabriel and Schmitzer, Bernhard and Vialard, Fran{\c{c}}ois-Xavier},
  journal={Journal of Functional Analysis},
  volume={274},
  number={11},
  pages={3090--3123},
  year={2018},
  publisher={Elsevier}
}

@article{liero2018optimal,
  title={Optimal entropy-transport problems and a new Hellinger--Kantorovich distance between positive measures},
  author={Liero, Matthias and Mielke, Alexander and Savar{\'e}, Giuseppe},
  journal={Inventiones mathematicae},
  volume={211},
  number={3},
  pages={969--1117},
  year={2018},
  publisher={Springer}
}

@article{sejourne2023unbalanced,
  title={Unbalanced optimal transport, from theory to numerics},
  author={S{\'e}journ{\'e}, Thibault and Peyr{\'e}, Gabriel and Vialard, Fran{\c{c}}ois-Xavier},
  journal={Handbook of Numerical Analysis},
  volume={24},
  pages={407--471},
  year={2023},
  publisher={Elsevier}
}

@article{jordan1998variational,
  title={The variational formulation of the Fokker--Planck equation},
  author={Jordan, Richard and Kinderlehrer, David and Otto, Felix},
  journal={SIAM journal on mathematical analysis},
  volume={29},
  number={1},
  pages={1--17},
  year={1998},
  publisher={SIAM}
}

@article{santambrogio2015optimal,
  title={Optimal transport for applied mathematicians},
  author={Santambrogio, Filippo},
  journal={Birk{\"a}user, NY},
  volume={55},
  number={58-63},
  pages={94},
  year={2015},
  publisher={Springer}
}

@book{villani2021topics,
  title={Topics in optimal transportation},
  author={Villani, C{\'e}dric},
  volume={58},
  year={2021},
  publisher={American Mathematical Society}
}

@article{peyre2019computational,
  title={Computational optimal transport: With applications to data science},
  author={Peyr{\'e}, Gabriel and Cuturi, Marco and others},
  journal={Foundations and Trends{\textregistered} in Machine Learning},
  volume={11},
  number={5-6},
  pages={355--607},
  year={2019},
  publisher={Now Publishers, Inc.}
}

@article{holland1988causal,
  title={Causal inference, path analysis and recursive structural equations models},
  author={Holland, Paul W},
  journal={ETS Research Report Series},
  volume={1988},
  number={1},
  pages={i--50},
  year={1988},
  publisher={Wiley Online Library}
}

@article{pearl2009survey,
author = {Judea Pearl},
title = {{Causal inference in statistics: An overview}},
volume = {3},
journal = {Statistics Surveys},
publisher = {Amer. Statist. Assoc., the Bernoulli Soc., the Inst. Math. Statist., and the Statist. Soc. Canada},
pages = {96 -- 146},
year = {2009}
}

@book{rachev2006mass1,
  title={Mass Transportation Problems: Volume 1: Theory},
  author={Rachev, Svetlozar T and R{\"u}schendorf, Ludger},
  year={2006},
  publisher={Springer Science \& Business Media}
}

@article{heckman1995assessing,
  title={Assessing the case for social experiments},
  author={Heckman, James J and Smith, Jeffrey A},
  journal={Journal of economic perspectives},
  volume={9},
  number={2},
  pages={85--110},
  year={1995},
  publisher={American Economic Association}
}

@article{cambanis1976inequalities,
  title={Inequalities for $E k (x, y)$ when the marginals are fixed},
  author={Cambanis, Stamatis and Simons, Gordon and Stout, William},
  journal={Zeitschrift f{\"u}r Wahrscheinlichkeitstheorie und verwandte Gebiete},
  volume={36},
  number={4},
  pages={285--294},
  year={1976},
  publisher={Springer}
}

@article{robins1989analysis,
  title={The analysis of randomized and non-randomized AIDS treatment trials using a new approach to causal inference in longitudinal studies},
  author={Robins, James M},
  journal={Health service research methodology: a focus on AIDS},
  pages={113--159},
  year={1989},
  publisher={US Public Health Service}
}

@article{chernozhukov2021identification,
  title={Identification of hedonic equilibrium and nonseparable simultaneous equations},
  author={Chernozhukov, Victor and Galichon, Alfred and Henry, Marc and Pass, Brendan},
  journal={Journal of Political Economy},
  volume={129},
  number={3},
  pages={842--870},
  year={2021},
  publisher={The University of Chicago Press Chicago, IL}
}

@article{chernozhukov2017monge,
  title={Monge-kantorovich depth, quantiles, ranks and signs},
  author={Chernozhukov, Victor and Galichon, Alfred and Hallin, Marc and Henry, Marc},
  journal={Annals of Statistics},
  volume={45},
  number={1},
  pages={223--256},
  year={2017},
  publisher={Institute of Mathematical Statistics}
}

@article{del2024nonparametric,
  title={Nonparametric multiple-output center-outward quantile regression},
  author={del Barrio, Eustasio and Sanz, Alberto Gonzalez and Hallin, Marc},
  journal={Journal of the American Statistical Association},
  pages={1--15},
  year={2024},
  publisher={Taylor \& Francis}
}

@article{fan2022lorenz,
  title={Lorenz map, inequality ordering and curves based on multidimensional rearrangements},
  author={Fan, Yanqin and Henry, Marc and Pass, Brendan and Rivero, Jorge A},
  journal={arXiv preprint arXiv:2203.09000},
  year={2022}
}

@article{heckman2001micro,
  title={Micro data, heterogeneity, and the evaluation of public policy: Nobel lecture},
  author={Heckman, James J},
  journal={Journal of political Economy},
  volume={109},
  number={4},
  pages={673--748},
  year={2001},
  publisher={The University of Chicago Press}
}

@article{balke1997bounds,
  title={Bounds on treatment effects from studies with imperfect compliance},
  author={Balke, Alexander and Pearl, Judea},
  journal={Journal of the American statistical Association},
  volume={92},
  number={439},
  pages={1171--1176},
  year={1997},
  publisher={Taylor \& Francis}
}

@article{imbens2007nonadditive,
  title={Nonadditive models with endogenous regressors},
  author={Imbens, Guido W},
  journal={Econometric Society Monographs},
  volume={43},
  pages={17},
  year={2007},
  publisher={Cambridge University Press}
}

@article{heckman1990varieties,
  title={Varieties of selection bias},
  author={Heckman, James},
  journal={The American Economic Review},
  volume={80},
  number={2},
  pages={313--318},
  year={1990},
  publisher={JSTOR}
}

@book{kallenberg1997foundations,
  title={Foundations of modern probability},
  author={Kallenberg, Olav and Kallenberg, Olav},
  volume={2},
  year={1997},
  publisher={Springer}
}

@article{gunsilius2023condition,
  title={A condition for the identification of multivariate models with binary instruments},
  author={Gunsilius, Florian F},
  journal={Journal of Econometrics},
  volume={235},
  number={1},
  pages={220--238},
  year={2023},
  publisher={Elsevier}
}

@article{d2015identification,
  title={Identification of nonseparable triangular models with discrete instruments},
  author={D'Haultf{\oe}uille, Xavier and F{\'e}vrier, Philippe},
  journal={Econometrica},
  volume={83},
  number={3},
  pages={1199--1210},
  year={2015},
  publisher={Wiley Online Library}
}

@article{chiappori2017multi,
  title={Multi-to One-Dimensional Optimal Transport},
  author={Chiappori, Pierre-Andr{\'e} and McCann, Robert J and Pass, Brendan},
  journal={Communications on Pure and Applied Mathematics},
  volume={70},
  number={12},
  pages={2405--2444},
  year={2017},
  publisher={Wiley Online Library}
}

@article{mccann2020optimal,
  title={Optimal transportation between unequal dimensions},
  author={McCann, Robert J and Pass, Brendan},
  journal={Archive for Rational Mechanics and Analysis},
  volume={238},
  number={3},
  pages={1475--1520},
  year={2020},
  publisher={Springer}
}

@article{hoderlein2007identification,
  title={Identification of marginal effects in nonseparable models without monotonicity},
  author={Hoderlein, Stefan and Mammen, Enno},
  journal={Econometrica},
  volume={75},
  number={5},
  pages={1513--1518},
  year={2007},
  publisher={Wiley Online Library}
}

@article{bound1995problems,
  title={Problems with instrumental variables estimation when the correlation between the instruments and the endogenous explanatory variable is weak},
  author={Bound, John and Jaeger, David A and Baker, Regina M},
  journal={Journal of the American statistical association},
  volume={90},
  number={430},
  pages={443--450},
  year={1995},
  publisher={Taylor \& Francis}
}

@article{gunsilius2021nontestability,
  title={Nontestability of instrument validity under continuous treatments},
  author={Gunsilius, Florian F},
  journal={Biometrika},
  volume={108},
  number={4},
  pages={989--995},
  year={2021},
  publisher={Oxford University Press}
}

@book{mikami2021stochastic,
  title={Stochastic optimal transportation: stochastic control with fixed marginals},
  author={Mikami, Toshio},
  year={2021},
  publisher={Springer Nature}
}

@inproceedings{pearl1995testability,
  title={On the testability of causal models with latent and instrumental variables},
  author={Pearl, Judea},
  booktitle={Proceedings of the Eleventh conference on Uncertainty in artificial intelligence},
  pages={435--443},
  year={1995}
}

@article{stock2002survey,
  title={A survey of weak instruments and weak identification in generalized method of moments},
  author={Stock, James H and Wright, Jonathan H and Yogo, Motohiro},
  journal={Journal of Business \& Economic Statistics},
  volume={20},
  number={4},
  pages={518--529},
  year={2002},
  publisher={Taylor \& Francis}
}

@article{carlier2016vector,
  title={Vector Quantile Regression: An Optimal Transport Approach},
  author={Carlier, Guillaume and Chernozhukov, Victor and Galichon, Alfred},
  journal={Annals of Statistics},
  volume={44},
  number={3},
  pages={1165--1192},
  year={2016}
}

@article{robins1992estimation,
  title={Estimation of the time-dependent accelerated failure time model in the presence of confounding factors},
  author={Robins, James},
  journal={Biometrika},
  volume={79},
  number={2},
  pages={321--334},
  year={1992},
  publisher={Oxford University Press}
}

@article{de2024transport,
  title={Transport-based counterfactual models},
  author={De Lara, Lucas and Gonz{\'a}lez-Sanz, Alberto and Asher, Nicholas and Risser, Laurent and Loubes, Jean-Michel},
  journal={Journal of Machine Learning Research},
  volume={25},
  number={136},
  pages={1--59},
  year={2024}
}

@article{brenier1991polar,
  title={Polar factorization and monotone rearrangement of vector-valued functions},
  author={Brenier, Yann},
  journal={Communications on pure and applied mathematics},
  volume={44},
  number={4},
  pages={375--417},
  year={1991},
  publisher={Wiley Online Library}
}

@book{rockafellar1970convex,
  author = {Rockafellar, R. Tyrrell},
  publisher = {Princeton University Press},
  series = {Princeton Mathematical Series},
  title = {Convex analysis},
  year = 1970
}

@article{gunsilius2023distributional,
  title={Distributional synthetic controls},
  author={Gunsilius, Florian F},
  journal={Econometrica},
  volume={91},
  number={3},
  pages={1105--1117},
  year={2023},
  publisher={Wiley Online Library}
}

@article{vansteelandt2014structural,
author = {Stijn Vansteelandt and Marshall Joffe},
title = {{Structural Nested Models and G-estimation: The Partially Realized Promise}},
volume = {29},
journal = {Statistical Science},
number = {4},
publisher = {Institute of Mathematical Statistics},
pages = {707 -- 731},
year = {2014}
}

@book{laan2003unified,
  title={Unified methods for censored longitudinal data and causality},
  author={{van der Laan}, Mark J and Robins, James M},
  year={2003},
  publisher={Springer}
}

@article{aronow2014sharp,
  title={Sharp bounds on the variance in randomized experiments},
  author={Aronow, Peter M and Green, Donald P and Lee, Donald KK},
  journal={The Annals of Statistics},
  pages={850--871},
  year={2014},
  publisher={JSTOR}
}

@article{balakrishnan2023conservative,
  title={Conservative inference for counterfactuals},
  author={Balakrishnan, Sivaraman and Kennedy, Edward and Wasserman, Larry},
  journal={arXiv preprint arXiv:2310.12757},
  year={2023}
}

@article{fan2010sharp,
  title={Sharp bounds on the distribution of treatment effects and their statistical inference},
  author={Fan, Yanqin and Park, Sang Soo},
  journal={Econometric Theory},
  volume={26},
  number={3},
  pages={931--951},
  year={2010},
  publisher={Cambridge University Press}
}

@article{mccann1999exact,
  title={Exact solutions to the transportation problem on the line},
  author={McCann, Robert J},
  journal={Proceedings of the Royal Society of London. Series A: Mathematical, Physical and Engineering Sciences},
  volume={455},
  number={1984},
  pages={1341--1380},
  year={1999},
  publisher={The Royal Society}
}

@inproceedings{balke1994counterfactual,
  title={Counterfactual probabilities: Computational methods, bounds and applications},
  author={Balke, Alexander and Pearl, Judea},
  booktitle={Uncertainty in artificial intelligence},
  pages={46--54},
  year={1994},
  organization={Elsevier}
}

@article{rubin1974estimating,
  title={Estimating causal effects of treatments in randomized and nonrandomized studies.},
  author={Rubin, Donald B},
  journal={Journal of educational Psychology},
  volume={66},
  number={5},
  pages={688},
  year={1974},
  publisher={American Psychological Association}
}


\end{document}